\documentclass[11pt]{amsart}

\usepackage{amsbsy,url}
\usepackage{amssymb,latexsym,amsmath,amsthm,amsfonts}
\usepackage{color} 
\usepackage{float}
\usepackage{placeins}
\usepackage{graphicx}
\usepackage{subcaption}
\usepackage{tikz}
\usetikzlibrary{arrows.meta,positioning}
\theoremstyle{plain}
\newcounter{Theorem}
\newtheorem{theorem}[Theorem]{Theorem}
\newtheorem{lemma}[Theorem]{Lemma} \newtheorem{proposition}[Theorem]{Proposition}
\newtheorem{corollary}[Theorem]{Corollary} \theoremstyle{remark} \newcounter{Definition}
\newcounter{Remark}
\newtheorem{definition}[Definition]{Definition}
\newtheorem{remark}[Remark]{Remark}

\theoremstyle{plain}
\newtheorem{theo}{\bf Theorem}[section]

\theoremstyle{definition}

\newtheorem{example}[theo]{\bf Example}
\theoremstyle{remark}

\numberwithin{equation}{section}
\makeatletter
\@namedef{subjclassname@2020}{%
  \textup{2020} Mathematics Subject Classification}
\makeatother

\begin{document}
\title{Enhancing Phase Retrievability of Quantum Channels via Interferometric Coupling}

\author{Kai Liu}
\address{Department of Mathematics, University of Central Florida, Orlando, FL 32816, USA}
\email{kai.liu@ucf.edu}
\thanks{\textsuperscript{*}Corresponding author: \texttt{kai.liu@ucf.edu}.}

\author{Deguang Han}
\address{Department of Mathematics, University of Central Florida, Orlando, FL 32816, USA}
\email{deguang.han@ucf.edu}

\author{Omar Nour}
\address{Department of Mathematics, University of Central Florida, Orlando, FL 32816, USA}
\email{omar.nour@ucf.edu}



\thispagestyle{plain}
\pagestyle{plain}





\begin{abstract}
Phase retrievability of a quantum channel asks whether pure states can be reconstructed from suitable measurements. In this paper, we study this problem from three complementary viewpoints: quantum information theory, operator-valued frames, and the physical realization through quantum interferometry. We first show that a quantum channel is phase retrievable if and only if its complementary channel is pure-state informationally complete. This structural characterization leads to several consequences for phase retrievability, including criteria involving the dimension of the complementary operator system, Choi-rank type bounds, and specific results for entanglement breaking channels and twirling channels.\\
\\ 
We then introduce an interferometric coupling in which two arm channels are coherently recombined through port operators \(M_i(\theta)=A_i+e^{i\theta}B_i\). Unlike classical mixing, this construction produces interference cross terms that can enlarge the complementary operator system and thereby enhance phase retrievability. From the frame theory viewpoint, the interferometer realizes a coherent coupling of operator-valued frames. To quantify this effect, we introduce injectivity indices for completely positive maps. The examples in Section~5 show that coherent interference can significantly improve phase retrieval behavior even when the arm channels are individually not phase retrievable.
\end{abstract}

\maketitle

\section{Introduction}

\noindent Phase retrieval is a classical inverse problem that asks whether a vector can be reconstructed, up to a unimodular scalar, from phaseless measurement data\cite{BCE06}. In quantum physics, closely related questions arise in pure-state tomography, where one asks whether a measurement scheme determines an unknown pure state uniquely and how many observables are required for such a reconstruction. This viewpoint appears naturally in tomography under prior information \cite{HMW13}, in measurement schemes using few orthonormal bases \cite{ZambranoPereiraDelgado2024}, and more recently in shadow-based approaches that seek efficient reconstruction or prediction from randomized measurements \cite{HuangKuengPreskill2020,GrierPashayanSchaeffer2024,QinLukensKirbyZhu2025}. These developments highlight that pure state distinguishability is both a structural and a physically meaningful problem, sitting at the interface of quantum measurement, tomography, and information extraction.\\
\\
A quantum channel describes the evolution of a quantum state in the presence of noise or environmental interaction. Once a pure state passes through such a channel, some information may be lost, redistributed, or transferred to the environment. It is therefore natural to ask whether the channel output still contains enough information to recover the input pure state, at least up to its physically irrelevant global phase. This leads to the notion of phase retrievability for quantum channels. Recently, this notion was introduced and studied in \cite{LiuHan2025}, and it was further examined for twirling channels in \cite{HanLiu2024}. Related recent work on quantum process tomography and channel learning has also emphasized the broader problem of reconstructing or characterizing unknown quantum channels from limited access and noisy data \cite{MeleEtAl2025,ChenYuZhang2025,GrossEtAl2025,JayakumarEtAl2024}. \\
\\
One guiding principle of the present paper is that phase retrievability can be reformulated structurally in terms of complementary channels. In quantum information theory, many properties of a channel become more transparent when viewed from the environmental side of the evolution, as captured by its complementary channels. This viewpoint has proved useful in several related settings, including the study of mixed unitary channels and the structure of reversible quantum channels \cite{GirardEtAl2022,LiLuoSun2026}. From this perspective, our first main result shows that phase retrievability is equivalent to pure-state informational completeness of a complementary channel. Thus the problem may be studied through the operator system generated by complementary Kraus products.\\
\\
The second guiding principle is that the problem also admits a natural interpretation in terms of operator-valued frames. Frame theory gives a convenient language for describing the relation among operator families, POVMs, and structured quantum channels \cite{KaftalLarsonZhang2009,HanLiMengTang2011,FDTT2024,HanJuste2019}. In our setting, both the Kraus operators and the pullbacks of output POVMs through the adjoint channel naturally live on the input space. This makes it possible to study phase retrievability not only at the level of state evolution, but also through the associated operator-valued measurement structure.\\
\\
 The novelty of this paper is that we do not stop at a static characterization. Instead, we introduce a physically motivated mechanism for enhancing phase retrievability by means of a two-arm quantum interferometer. Interferometric studies of quantum processes have shown that coherence information about completely positive maps can be accessed through path interference \cite{Oi03,OiAberg2006}. Motivated by this viewpoint and some recent work \cite{TanRohde2019, AbbottWechsHorsmanMhallaBranciard2020, BranciardEtAl2021}, we consider two arm channels and their coherent recombination inside an interferometer, which at the Kraus level leads to operators of the form
\begin{equation}
M_i(\theta)=A_i+e^{i\theta}B_i.
\end{equation}
Mathematically, this may be viewed as a coherent coupling of operator-valued frames, and the resulting interference cross terms provide a natural mechanism for enlarging the complementary operator system.\\
\\
More precisely, the paper has four main contributions. First, we characterize phase retrievability of quantum channels in terms of complementary channels and pure state informational completeness. Second, we derive necessary conditions for phase retrievability in terms of complementary operator system dimension, and discuss their consequences for entanglement breaking and twirling channels. Third, we introduce an operator-valued frame coupling as the underlying mathematical mechanism, and show that quantum interference in a two-arm interferometer provides a natural physical realization of this construction. Finally, we introduce a quantitative index and use it to study explicit examples illustrating the enhancement phenomenon.\\
\\

\medskip\noindent
\textbf{Organization:}\\
\noindent The remainder of the paper is organized as follows. Section~2 collects the necessary background on 
quantum channels, phase retrievability, operator-valued frames, and quantum interferometry. 
Section~3 develops the complementary channel characterization of phase retrievability and several 
structural consequences. Section~4 introduces the quantum interferometric coupling, and explains its 
theoretical and physical significance. Section~5 presents a quantitative injectivity index together with 
illustrative examples. Section~6 concludes with a summary and outlook.

\section{Preliminaries}

\noindent Throughout this paper, all Hilbert spaces are complex finite-dimensional spaces. 
For Hilbert spaces $\mathcal{H}$ and $\mathcal{K}$, we write $B(\mathcal{H},\mathcal{K})$ 
as the space of all linear operators from $\mathcal{H}$ to $\mathcal{K}$, and 
$B(\mathcal{H}) = B(\mathcal{H},\mathcal{H})$. The space of all self-adjoint operators on $\mathcal{H}$ is denoted by $B_{sa}(\mathcal{H})$. The Hilbert-Schmidt inner product on 
$B(\mathcal{H})$ is given by
\begin{equation}
\langle X,Y\rangle = \operatorname{Tr}(Y^{*}X), \qquad X,Y\in B(\mathcal{H}).
\end{equation}

\subsection{Quantum channels and phase retrievability}

A quantum state on $\mathcal{H}$ is a positive semidefinite operator $\rho\in B(\mathcal{H})$ satisfying 
$\operatorname{Tr}(\rho)=1$. A quantum state is called \emph{pure state} if it is an extreme point of the convex 
set of all states. In finite dimensions, this is equivalent to saying that $\rho$ is a rank-one 
orthogonal projection. Thus every pure state has the form $\rho_x = |x\rangle\langle x|$ for some unit vector $x\in \mathcal{H}$.\\
\\
A quantum channel $\Phi : B(\mathcal{H}) \to B(\mathcal{K})$ is a completely positive 
trace preserving (CPTP) map. The adjoint map of $\Phi$ is the linear map $\Phi^{*}:B(\mathcal{K})\to B(\mathcal{H})$ defined by
\begin{equation}
\langle \Phi(X),Y\rangle = \langle X,\Phi^{*}(Y)\rangle,
\qquad X\in B(\mathcal{H}), \ Y\in B(\mathcal{K}),
\end{equation}
with respect to the Hilbert-Schmidt inner product. A linear map 
$\Psi:B(\mathcal{K})\to B(\mathcal{H})$ is called \emph{unital} if $\Psi(I_{\mathcal K}) = I_{\mathcal H}$. Immediately, if $\Phi$ is trace preserving, then its adjoint map $\Phi^{*}$ is unital.\\
\\
Every completely positive map admits a Kraus representation. In particular, if 
$\Phi : B(\mathcal{H}) \to B(\mathcal{K})$ is a quantum channel, then there exist operators 
$A_1,\dots,A_m \in B(\mathcal{H},\mathcal{K})$ such that
\begin{equation}
\Phi(\rho)=\sum_{i=1}^{m} A_i \rho A_i^{*}, \qquad \rho\in B(\mathcal{H}).
\end{equation}
Moreover, $\Phi$ is trace preserving if and only if $\sum_{i=1}^{m} A_i^{*}A_i = I_{\mathcal H}$, and the adjoint map $\Phi^*$ is given by $\Phi^*(\rho)=\sum_{i=1}^{m} A_i^* \rho A_i$. \\
\\
Different Kraus representations of the same channel are obtained from one another by unitary mixing of the Kraus operators. More precisely, it is described by following standard result.

\begin{theorem}\cite{NC00}\label{kraus theorem}
Let $\Phi : B(\mathcal{H}) \to B(\mathcal{K})$ be a quantum channel, and suppose that $\Phi(\rho)=\sum_{i=1}^{m} A_i \rho A_i^{*}$ and $\Phi(\rho)=\sum_{j=1}^{n} B_j \rho B_j^{*}$ are two Kraus representations of $\Phi$. Then, after padding the shorter list with zero operators so that $m=n$, there exists a unitary matrix $U=(u_{ij})\in M_n(\mathbb{C})$ such that $A_i=\sum_{j=1}^{n} u_{ij} B_j, \, i=1,\dots,n.$
\end{theorem}

\noindent A quantum channel may equivalently be described by a Stinespring dilation. Indeed, for every 
quantum channel $\Phi:B(\mathcal{H})\to B(\mathcal{K})$, there exist a Hilbert space $\mathcal{Z}$ 
and an isometry $V:\mathcal{H}\to \mathcal{K}\otimes\mathcal{Z}$ such that
\begin{equation}
\Phi(\rho)=\operatorname{Tr}_{\mathcal{Z}}(V\rho V^{*}), \qquad \rho\in B(\mathcal{H}).
\end{equation}
The associated complementary channel is then defined by tracing out the output space:
\begin{equation}
\Phi^{c}(\rho)=\operatorname{Tr}_{\mathcal{K}}(V\rho V^{*}), \qquad \rho\in B(\mathcal{H}).
\end{equation}
Thus $\Phi$ and $\Phi^{c}$ are obtained from the same isometric evolution by tracing out different 
subsystems. The channel $\Phi$ describes the output seen in $\mathcal{K}$, while $\Phi^{c}$ records 
the information transferred to the environment $\mathcal{Z}$. Any two complementary channels 
arising from Stinespring dilations of the same channel are unitarily equivalent on the environment 
space.\\
\\
A positive operator-valued measure (POVM) on $\mathcal{K}$ is a finite set of positive 
operators $\{F_j\}_{j=1}^{N}\subset B(\mathcal{K})$ such that $\sum_{j=1}^{N} F_j = I_{\mathcal K}$. Let $\Phi:B(\mathcal{H})\to B(\mathcal{K})$ be a quantum channel, then for any POVM 
$\{F_j\}_{j=1}^{N}$ on $\mathcal{K}$, the set $\{\Phi^{*}(F_j)\}_{j=1}^{N}$ is again a POVM on 
$\mathcal{H}$, since $\Phi^{*}$ is unital. We now recall the notion of phase retrievability for 
quantum channels.

\begin{definition}\cite{LiuHan2025}
A quantum channel $\Phi:B(\mathcal{H})\to B(\mathcal{K})$ is called \emph{phase retrievable} if 
there exists a POVM $\{F_j\}_{j=1}^{N}$ on $\mathcal{K}$ such that $\operatorname{Tr}(\rho_x\,\Phi^{*}(F_j))
=\operatorname{Tr}(\rho_y\,\Phi^{*}(F_j)),
\; j=1,\dots,N,$ implies $\rho_x=\rho_y$ for all unit vectors $x,y\in \mathcal{H}$.
\end{definition}

\subsection{Operator-valued frames and their connection with quantum channels}

We briefly recall the definition of operator-valued frames. A finite sequence 
$\mathcal{A}=\{A_i\}_{i=1}^{m}\subset B(\mathcal{H},\mathcal{K})$ is called an 
\emph{operator-valued frame} for $\mathcal{H}$ if there exist constants $\alpha,\beta>0$ such that
\begin{equation}
\alpha \|x\|^{2}
\leq
\sum_{i=1}^{m} \|A_i x\|^{2}
\leq
\beta \|x\|^{2},
\qquad x\in \mathcal{H}.
\end{equation}
Equivalently,
\begin{equation}
\alpha I_{\mathcal H}
\leq
\sum_{i=1}^{m} A_i^{*}A_i
\leq
\beta I_{\mathcal H}.
\end{equation}
The operator $S_{\mathcal A}=\sum_{i=1}^{m} A_i^{*}A_i$ is called the \emph{frame operator} of $\mathcal{A}$. The sequence $\mathcal{A}$ is called 
\emph{tight} if $S_{\mathcal A}=cI_{\mathcal H}$ for some $c>0$, and \emph{Parseval} if $S_{\mathcal A}=I_{\mathcal H}$.\\
\\
Further, a sequence $\mathcal{A}=\{A_i\}_{i=1}^{m}\subset B(\mathcal{H},\mathcal{K})$ is called 
\emph{phase retrievable} if $\|A_i x\|=\|A_i y\|, \; i=1,\dots,m,$ implies $\rho_x=\rho_y$ for all unit vectors $x,y\in \mathcal{H}$. Since
\begin{equation}
\|A_i x\|^{2}=\langle A_i^{*}A_i x,x\rangle,
\end{equation}
this is equivalent to saying that the positive operators $\{A_i^{*}A_i\}_{i=1}^{m}$ separate pure 
states on $\mathcal{H}$. In the special case $\mathcal{K}=\mathbb{C}$, this reduces to the usual 
notion of phase retrieval for vector frames.\\
\\
The framework of operator-valued frames is naturally related to quantum channels. Indeed, If $\Phi$ is a quantum channel with Kraus representation $\Phi(\rho)=\sum_{i=1}^{m} A_i \rho A_i^{*}$, then we have $\sum_{i=1}^{m} A_i^{*}A_i=I_{\mathcal H}$ since $\Phi$ is trace preserving. Hence every Kraus sequence of a quantum channel is a Parseval operator-valued frame.\\
\\
On the other hand, let $\Phi:B(\mathcal{H})\to B(\mathcal{K})$ be a quantum channel and let 
$\{F_i\}_{i=1}^{N}$ be a POVM on $\mathcal{K}$. Since $\Phi^{*}$ is unital, the sequence 
$\{\Phi^{*}(F_i)\}_{i=1}^{N}$ is a POVM on $\mathcal{H}$. Consequently, 
$\{\Phi^{*}(F_i)^{1/2}\}_{i=1}^{N}$ is a Parseval operator-valued frame for $\mathcal{H}$. Moreover,
\begin{equation}
\|\Phi^{*}(F_i)^{1/2}x\|^{2}
=
\langle \Phi^{*}(F_i)x,x\rangle
=
\operatorname{Tr}(\rho_x\,\Phi^{*}(F_i)).
\end{equation}
Therefore $\Phi$ is phase retrievable if and only if there exists a POVM 
$\{F_i\}_{i=1}^{N}$ on $\mathcal{K}$ such that the Parseval operator-valued frame 
$\{\Phi^{*}(F_i)^{1/2}\}_{i=1}^{N}$ is phase retrievable on $\mathcal{H}$.

\subsection{Quantum interferometers and interference of completely positive maps}

We now recall the interferometric background that motivates the discussion in Section 4. A 
quantum interferometer is a device that coherently splits an incoming wave into different paths, 
allows the amplitudes in those paths to evolve separately, and then recombines them. Because the 
recombination takes place at the amplitude level, the output intensities depend on the relative phase 
between the paths. This phase-dependent behavior is called \emph{interference}. In quantum theory, 
interferometers are standard tools for probing coherence properties of states and processes. A 
familiar example is the two-path Mach-Zehnder interferometer.\\
\\
A schematic two-path interferometer is shown in Figure~\ref{fig:interferometer}. After the first 
beam splitter (BS), the input amplitude is divided coherently into a lower arm and an upper arm. The two 
amplitudes then evolve separately, a relative phase may be inserted in one arm, and the paths are 
finally recombined at the second beam splitter. The output ports record the resulting interference 
pattern.

\begin{figure}[htbp]
\centering
\begin{tikzpicture}[>=Latex, thick, every node/.style={font=\small}]

\draw[->] (-2,0) -- (-1.1,0);
\node[left] at (-2,0) {$|\psi\rangle$};

\draw (-1.1,-0.55) rectangle (-0.8,0.55);
\node[rotate=90] at (-0.95,0) {\small BS};

\draw[->] (-0.8,0.32) -- (0.2,1.4);
\draw[->] (-0.8,-0.32) -- (1.2,-1.4);

\draw (0.4,1.05) rectangle (1.8,1.75);
\node at (1.1,1.4) {$\Phi_B$};

\draw (2.2,1.05) rectangle (3.5,1.75);
\node at (2.85,1.4) {$e^{i\theta}$};

\draw (1.4,-1.75) rectangle (2.7,-1.05);
\node at (2,-1.4) {$\Phi_A$};

\draw[->] (1.8,1.4) -- (2.2,1.4);
\draw[->] (3.5,1.4) -- (4.7,0.32);
\draw[->] (2.7,-1.4) -- (4.7,-0.32);

\draw (4.7,-0.55) rectangle (5.0,0.55);
\node[rotate=90] at (4.85,0) {\small BS};

\draw[->] (5.0,0.32) -- (6.2,1.15);
\draw[->] (5.0,-0.32) -- (6.2,-1.15);

\node[right] at (6.2,1.15) {port 0};
\node[right] at (6.2,-1.15) {port 1};

\end{tikzpicture}
\caption{A two-path quantum interferometer. The lower arm carries the process $\Phi_A$, the upper 
arm carries the process $\Phi_B$, and a relative phase $e^{i\theta}$ is inserted in the upper arm 
before recombination.}
\label{fig:interferometer}
\end{figure}
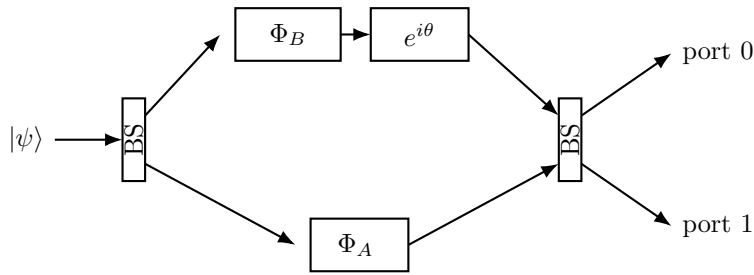

\noindent To describe the simplest case mathematically, let the path degree of freedom be represented by the 
orthonormal basis $\{|0\rangle,|1\rangle\}$, where $|0\rangle$ corresponds to the lower arm and 
$|1\rangle$ corresponds to the upper arm. If the internal state is $|\psi\rangle\in\mathcal{H}$ and 
the two arms carry unitary evolutions $U_A$ and $U_B$, then after inserting a relative phase 
$e^{i\theta}$ in the upper arm, the state before the second beam splitter has the schematic form
\begin{equation}
\frac{1}{\sqrt{2}}
\left(
|0\rangle\otimes U_A|\psi\rangle
+
e^{i\theta}|1\rangle\otimes U_B|\psi\rangle
\right).
\end{equation}
After recombination, the probability of exiting through one output port is
\begin{equation}
P_{0}(\theta)
=
\frac{1}{2}\bigl(1+v\cos(\theta-\alpha)\bigr),
\end{equation}
where $ve^{i\alpha}=\langle \psi|U_A^{*}U_B|\psi\rangle$.\\
\\
Thus the interference pattern is determined by the overlap of the two arm amplitudes. For a mixed 
input state $\rho$, the same quantity is $ve^{i\alpha}=\operatorname{Tr}(\rho\,U_A^{*}U_B)$.\\
\\
This picture can be extended to general quantum processes. In the two arms of an interferometer, the 
evolution need not be unitary on the internal system alone, since the system may interact with an 
environment. Daniel K. L. Oi showed that the interference of two quantum channels can be modeled by 
choosing coherent Stinespring implementations for the channels acting in the two arms \cite{Oi03}. 
Let the lower arm carry the channel $\Phi_A:B(\mathcal{H})\to B(\mathcal{K})$ and the upper arm 
carry the channel $\Phi_B:B(\mathcal{H})\to B(\mathcal{K})$. Choose Hilbert spaces 
$\mathcal{E}_A,\mathcal{E}_B$ and isometries
\begin{equation}
V_A:\mathcal{H}\to \mathcal{K}\otimes \mathcal{E}_A,
\qquad
V_B:\mathcal{H}\to \mathcal{K}\otimes \mathcal{E}_B
\end{equation}
such that
\begin{equation}
\Phi_A(\rho)=\operatorname{Tr}_{\mathcal E_A}(V_A\rho V_A^{*}),
\qquad
\Phi_B(\rho)=\operatorname{Tr}_{\mathcal E_B}(V_B\rho V_B^{*}).
\end{equation}
After choosing orthonormal bases in $\mathcal{E}_A$ and $\mathcal{E}_B$, the channels may be 
written in Kraus form as
\begin{equation}
\Phi_A(\rho)=\sum_{i=1}^{m} A_i \rho A_i^{*},
\qquad
\Phi_B(\rho)=\sum_{i=1}^{m} B_i \rho B_i^{*},
\end{equation}
where the shorter Kraus list is padded with zero operators if necessary.\\
\\
For a fixed coherent implementation, the interference depends on cross terms between the Kraus 
operators in the two arms. In particular, expressions of the form
\begin{equation}
\operatorname{Tr}(A_i^{*}B_i\,\rho)
\end{equation}
appear naturally in the visibility formula \cite{Oi03}. The important point for the present paper is 
that the resulting interference pattern depends not only on the input-output action of the channels, 
but also on the chosen coherent implementation, and hence on the Kraus realization.\\
\\
We use two facts below: interference produces cross terms, and the resulting map depends on the chosen coherent Kraus-operator implementation.

\section{Phase Retrievability and Complementary Operator Systems}

\subsection{Characterization of phase retrievability via complementary channels}

\noindent Our later discussion of operator-valued frame coupling is formulated at the level of Kraus families of quantum channels. For this reason, whenever needed, we will fix a Kraus representation of the given quantum channel and define the associated constructions relative to that choice.\\
\\
We first rewrite the definition of complementary channels via Kraus representations.
\begin{definition}
Let $\Phi_{\mathcal A}(\rho)=\sum_{i=1}^{r} A_i \rho A_i^{*}$ be a Kraus representation of a quantum channel $\Phi$ from $B(\mathcal{H})$ to $B(\mathcal{K})$ with Kraus operators $\mathcal A=\{A_i\}_{i=1}^{r}\subseteq B(\mathcal{H},\mathcal{K})$. The complementary channel associated with this Kraus representation $\mathcal A$ is given by
\begin{equation}
\Phi_{\mathcal A}^{c}(\rho)
=
\sum_{i,j=1}^{r} \operatorname{Tr}\!\left(A_i \rho A_j^{*}\right)\, |i\rangle\langle j|,
\end{equation}
where $\{|i\rangle\}_{i=1}^{r}$ is the standard orthonormal basis of $\mathbb{C}^{r}$.
\end{definition}

\begin{remark}\label{canonical}
For a fixed Kraus family $\mathcal A=\{A_i\}_{i=1}^{r}$, the complementary channel $\Phi_{\mathcal A}^{c}$ is well defined, but its Kraus representation is not unique by Theorem \ref{kraus theorem}. Fix an orthonormal basis $\{|a\rangle\}_{a=1}^{d_\mathcal{K}}$ of $\mathcal{K}$ with $d_{\mathcal{K}}=\dim \mathcal{K}$.
Then, it is straightforward to verify that
\begin{equation}
\Phi_{\mathcal A}^{c}(\rho)=\sum_{a=1}^{d_\mathcal{K}} R_a \rho R_a^{*},
\end{equation}
where $R_a=\sum_{i=1}^{r} |i\rangle \langle a| A_i$.\\
\\
We shall call $\mathcal R_{\mathcal A}:=\{R_a\}_{a=1}^{d_\mathcal{K}}$ the \emph{canonical Kraus representation} of $\Phi_{\mathcal A}^{c}$ associated with $\mathcal A$. Unless otherwise specified, all later discussions of Kraus operators of the complementary channel will be with respect to this canonical Kraus representation.
\end{remark}

\begin{definition}
Let $\mathcal E=\{E_{\alpha}\}_{\alpha\in\Lambda}\subseteq M_d$ be a family of operators. For all unit vectors $x, y \in \mathbb{C}^d$, we say that $\mathcal E$ is \emph{pure-state informationally complete} (PSIC) if $\langle E_{\alpha}x,x\rangle=\langle E_{\alpha}y,y\rangle \quad \text{for all } \alpha\in\Lambda$ implies $x=e^{i\theta}y$ for some $\theta\in\mathbb{R}$.
\end{definition}

\begin{definition}\label{PSIC channel}
A quantum channel $\Phi$ is called \emph{pure-state informationally complete} (PSIC) if there exists a Kraus representation $\Phi_{\mathcal A}(\rho)=\sum_{i=1}^{r} A_i \rho A_i^{*}$ with Kraus operators $\mathcal A=\{A_i\}_{i=1}^{r}$ such that the set $\{A_j^{*}A_i : 1\leq i,j\leq r\}$ is pure-state informationally complete.
\end{definition}

\noindent The notion of pure-state injectivity provides a convenient Schr\"odinger-picture formulation of phase retrievability. Following our previous work in \cite{LiuHan2025}, a quantum channel $\Phi : B(\mathcal{H})\to B(\mathcal{K})$ is said to be \emph{pure-state injective} if $\Phi(\rho_x)=\Phi(\rho_y)$ implies $\rho_x=\rho_y.$ The next proposition shows that this condition is equivalent to phase retrievability.

\begin{proposition}\label{pure injective}\cite{LiuHan2025}
A quantum channel $\Phi : B(\mathcal{H})\to B(\mathcal{K})$ is phase retrievable if and only if $\Phi$ is pure-state injective.
\end{proposition}

\noindent By Proposition~\ref{pure injective}, phase retrievability is equivalent to pure-state injectivity. We now investigate when a quantum channel could be phase retrievable in terms of its complementary channel associated with a fixed Kraus representation.

\begin{theorem}\label{PR iff PSIC}
Let $\Phi_{\mathcal A}(\rho)=\sum_{i=1}^{r} A_i \rho A_i^{*}$ be a Kraus representation of a quantum channel $\Phi$, and let $\Phi_{\mathcal A}^{c}$ be the complementary channel associated with $\mathcal A=\{A_i\}_{i=1}^{r}$. Then $\Phi$ is phase retrievable if and only if $\Phi_{\mathcal A}^{c}$ is pure-state informationally complete.
\end{theorem}

\begin{proof}
Let $\mathcal R_{\mathcal A}=\{R_a\}_{a=1}^{d_{\mathcal K}}$ be the canonical Kraus representation of $\Phi_{\mathcal A}^{c}$ associated with $\mathcal A$, where $R_a=\sum_{i=1}^{r} |i\rangle \langle a| A_i, \; 1\leq a\leq d_{\mathcal K},$ and $d_{\mathcal K}=\dim \mathcal K.$\\
\\
By Remark \ref{canonical}, we have $\Phi_{\mathcal A}^{c}(\rho)=\sum_{a=1}^{d_{\mathcal K}} R_a \rho R_a^{*}.$\\
\\
We first show that the complementary channel of $\Phi_{\mathcal A}^{c}$, taken with respect to the canonical Kraus representation $\mathcal R_{\mathcal A}$, recovers the original Kraus representation $\Phi_{\mathcal A}$. Indeed,
\begin{equation}\label{com of com}
\bigl(\Phi_{\mathcal A}^{c}\bigr)_{\mathcal R_{\mathcal A}}^{c}(\rho)
=
\sum_{a,b=1}^{d_{\mathcal K}} \operatorname{Tr}\!\left(R_a \rho R_b^{*}\right)\, |a\rangle\langle b|.
\end{equation}
With the definition of $R_a$, we have
\begin{equation}
R_a \rho R_b^{*}
=
\sum_{i,j=1}^{r} |i\rangle \langle a| A_i \rho A_j^{*} |b\rangle \langle j|.
\end{equation}
Hence
\begin{equation}
\operatorname{Tr}\!\left(R_a \rho R_b^{*}\right)
=
\sum_{i,j=1}^{r} \operatorname{Tr}\!\left(|i\rangle \langle a| A_i \rho A_j^{*} |b\rangle \langle j|\right)
=
\sum_{i=1}^{r} \langle a| A_i \rho A_i^{*} |b\rangle,
\end{equation}
and therefore
\begin{equation}\label{equ double com}
\bigl(\Phi_{\mathcal A}^{c}\bigr)_{\mathcal R_{\mathcal A}}^{c}(\rho)
=
\sum_{a,b=1}^{d_{\mathcal K}} \sum_{i=1}^{r} \langle a| A_i \rho A_i^{*} |b\rangle\, |a\rangle\langle b|
=
\sum_{i=1}^{r} A_i \rho A_i^{*}
=
\Phi_{\mathcal A}(\rho).
\end{equation}
Thus we have $\bigl(\Phi_{\mathcal A}^{c}\bigr)_{\mathcal R_{\mathcal A}}^{c}=\Phi_{\mathcal A}$.\\
\\
Now let $|x\rangle,|y\rangle \in \mathcal H$ be unit vectors and $\rho_x$, $\rho_y$ be the corresponding pure states. Then we have $\Phi_{\mathcal A}(\rho_x)=\Phi_{\mathcal A}(\rho_y)$ if and only if $\bigl(\Phi_{\mathcal A}^{c}\bigr)_{\mathcal R_{\mathcal A}}^{c}(\rho_x)=\bigl(\Phi_{\mathcal A}^{c}\bigr)_{\mathcal R_{\mathcal A}}^{c}(\rho_y).$\\
\\
By Equation\;\ref{com of com} and \ref{equ double com}, we have $\Phi_{\mathcal A}(\rho_x)=\Phi_{\mathcal A}(\rho_y)$ if and only if
\begin{equation}
\operatorname{Tr}\!\left(R_a \rho_x R_b^{*}\right)
=
\operatorname{Tr}\!\left(R_a \rho_y R_b^{*}\right)
\qquad \text{for all } 1\leq a,b\leq d_{\mathcal K}.
\end{equation}
Since $\operatorname{Tr}\!\left(R_a \rho_x R_b^{*}\right)=\langle x|R_b^{*}R_a|x\rangle$ and $\operatorname{Tr}\!\left(R_a \rho_y R_b^{*}\right)=\langle y|R_b^{*}R_a|y\rangle,$ we conclude that $\Phi_{\mathcal A}(\rho_x)=\Phi_{\mathcal A}(\rho_y)$ if and only if $\langle x|R_b^{*}R_a|x\rangle=\langle y|R_b^{*}R_a|y\rangle \; \text{for all } 1\leq a,b\leq d_{\mathcal K}.$\\
\\
Therefore, $\Phi_{\mathcal A}$ is pure-state injective if and only if the set $\{R_b^{*}R_a : 1\leq a,b\leq d_{\mathcal K}\}$ is pure-state informationally complete. By Definition \ref{PSIC channel}, this is equivalent to saying that $\Phi_{\mathcal A}^{c}$ is PSIC. Hence
\begin{equation}
\Phi_{\mathcal A} \text{ is pure-state injective }
\iff
\Phi_{\mathcal A}^{c} \text{ is PSIC.}
\end{equation}
Finally, Proposition~\ref{pure injective} yields
\begin{equation}
\Phi \text{ is phase retrievable }
\iff
\Phi_{\mathcal A}^{c} \text{ is PSIC},
\end{equation}
This completes the proof.
\end{proof}

\subsection{Complementary operator systems and necessary conditions for phase retrievability}
\noindent Theorem~\ref{PR iff PSIC} shows that phase retrievability of a quantum channel is determined by the pure-state informational completeness of its complementary channel. Consequently, general lower bounds for pure-state informational completeness yield necessary conditions for phase retrievability. In particular, we start with the following lower bound from \cite[Theorem 6]{HMW13}.

\begin{theorem}\label{lower bound HMW}\cite{HMW13}
Let $\mathfrak{m}$ denote the minimal number of self-adjoint operators that is informationally complete with respect to the set $\mathcal{P}_1$ of pure states in $\mathbb{C}^d$. Then informational completeness for pure states in $\mathbb{C}^d$ requires
\begin{equation}
\mathfrak{m}>
\begin{cases}
2D-2\alpha, & \text{for all } d>1,\\
2D-2\alpha+2, & \text{if } d \text{ is odd and } \alpha \equiv 3 \pmod{4},\\
2D-2\alpha+1, & \text{if } d \text{ is odd and } \alpha \equiv 2 \pmod{4},
\end{cases}
\end{equation}
where $\alpha$ denotes the number of $1$'s in the binary expansion of $d-1$, and $D=2d-2$ is the real dimension of the manifold of pure states.
\end{theorem}

\noindent In the following discussion, we shall denote the lower bound in Theorem~\ref{lower bound HMW} by $N(d)$. Thus any PSIC operator system must have dimension strictly larger than $N(d)$.\\
\\
\noindent For a quantum channel $\Phi$ with fixed Kraus representation $\Phi_{\mathcal A}(\rho)=\sum_{i=1}^{r} A_i \rho A_i^{*}$, we write $\Phi_{\mathcal A}^{*}(X)=\sum_{i=1}^{r} A_i^{*} X A_i$ for the adjoint map of $\Phi_{\mathcal A}$, where $X\in B(\mathcal K)$. We will also use the vectorization map $\operatorname{vec}:M_d\to \mathbb{C}^{d}\otimes \mathbb{C}^{d}$, which is defined by
\begin{equation}
\operatorname{vec}\!\left(\sum_{a,b=1}^{d} x_{ab}\,|a\rangle\langle b|\right)
=
\sum_{a,b=1}^{d} x_{ab}\,|b\rangle\otimes |a\rangle.
\end{equation}
The following lemma is standard. 

\begin{lemma}\label{vec identity}\cite{MagnusNeudecker1999}
Let $A$, $X$, and $B$ be matrices of compatible sizes. Then
\begin{equation}
\operatorname{vec}(AXB)=(B^{T}\otimes A)\operatorname{vec}(X).
\end{equation}
\end{lemma}

\noindent For the canonical Kraus representation $\mathcal R_{\mathcal A}=\{R_a\}_{a=1}^{d_{\mathcal K}}$ of $\Phi_{\mathcal A}^{c}$, define the associated operator system by $S_{\Phi_{\mathcal A}^{c}}:=\operatorname{span}\{R_b^{*}R_a : 1\leq a,b\leq d_{\mathcal K}\}.$ The following Lemma \ref{dim operator system} will be useful in subsequent discussions.

\begin{lemma}\label{dim operator system}
Let $\Phi_{\mathcal A}(\rho)=\sum_{i=1}^{r} A_i \rho A_i^{*}$ be a Kraus representation of a quantum channel $\Phi$. Then we have $S_{\Phi_{\mathcal A}^{c}}=\operatorname{Range}(\Phi_{\mathcal A}^{*}),$ and hence
\begin{equation}
\dim S_{\Phi_{\mathcal A}^{c}}=\dim \operatorname{Range}(\Phi_{\mathcal A}^{*})=\operatorname{rank}\left(\sum_{i=1}^{r} A_i^{T}\otimes A_i^{*}\right).
\end{equation}
\end{lemma}

\begin{proof}
Let $E_{ab}=|a\rangle\langle b|$, where $1\leq a,b\leq d_{\mathcal K}$ and $\{|a\rangle\}_{a=1}^{d_{\mathcal K}}$ is the fixed orthonormal basis of $\mathcal K$ used in the definition of the canonical Kraus representation $R_a=\sum_{i=1}^{r} |i\rangle\langle a| A_i.$\\
\\ 
Then
\begin{equation}
R_b^{*}R_a
=
\left(\sum_{i=1}^{r} A_i^{*}|b\rangle\langle i|\right)
\left(\sum_{j=1}^{r} |j\rangle\langle a| A_j\right)
=
\sum_{i=1}^{r} A_i^{*}|b\rangle\langle a|A_i
=
\Phi_{\mathcal A}^{*}(E_{ba}).
\end{equation}
Since $\{E_{ab}:1\leq a,b\leq d_{\mathcal K}\}$ is a basis of $B(\mathcal K)$, it follows that
\begin{equation}
S_{\Phi_{\mathcal A}^{c}}
=
\operatorname{span}\{R_b^{*}R_a : 1\leq a,b\leq d_{\mathcal K}\}
=
\operatorname{Range}(\Phi_{\mathcal A}^{*}).
\end{equation}
\\
Moreover, Lemma \ref{vec identity} implies that, for any $X \in B(\mathcal K)$,
\begin{equation}
\operatorname{vec}\!\left(\Phi_{\mathcal A}^{*}(X)\right)
=
\sum_{i=1}^{r} \operatorname{vec}(A_i^{*}X A_i)
=
\sum_{i=1}^{r} (A_i^{T}\otimes A_i^{*})\operatorname{vec}(X)
=
\left(\sum_{i=1}^{r} A_i^{T}\otimes A_i^{*}\right)\operatorname{vec}(X).
\end{equation}
Therefore,
\begin{equation}
\dim \operatorname{Range}(\Phi_{\mathcal A}^{*})
=
\operatorname{rank}\left(\sum_{i=1}^{r} A_i^{T}\otimes A_i^{*}\right),
\end{equation}
and the proof is complete.\end{proof}

\noindent By Theorem~\ref{PR iff PSIC}, Theorem~\ref{lower bound HMW}, and Lemma~\ref{dim operator system}, we immediately obtain the following observation: if
\begin{equation}
\operatorname{rank}\left(\sum_{i=1}^{r} A_i^{T}\otimes A_i^{*}\right)\leq N(d),
\end{equation}
then $\Phi$ is not phase retrievable. We now discuss several concrete consequences of this fact and begin with a general obstruction in terms of the ranks of the Kraus operators.

\begin{theorem}\label{general rank}
Let $\Phi : B(\mathcal H)\to B(\mathcal K)$ be a quantum channel and $d=\dim \mathcal H$. Suppose that $\Phi$ admits a Kraus representation $\{A_i\}_{i=1}^r$ with $\operatorname{rank}(A_i)=r_i,\; 1\leq i\leq r$. If $\sum_{i=1}^{r} r_i^{2}\leq N(d)$, then $\Phi$ is not phase retrievable.
\end{theorem}

\begin{proof}
By Theorem~\ref{PR iff PSIC}, Theorem~\ref{lower bound HMW}, and Lemma~\ref{dim operator system}, it suffices to show that
\begin{equation}
\operatorname{rank}\left(\sum_{i=1}^{r} A_i^{T}\otimes A_i^{*}\right)\leq N(d).
\end{equation}
Since rank is subadditive, we have
\begin{equation}
\operatorname{rank}\left(\sum_{i=1}^{r} A_i^{T}\otimes A_i^{*}\right)
\leq
\sum_{i=1}^{r} \operatorname{rank}\left(A_i^{T}\otimes A_i^{*}\right).
\end{equation}
Moreover, for each $i$,
\begin{equation}
\operatorname{rank}\left(A_i^{T}\otimes A_i^{*}\right)
=
\operatorname{rank}(A_i^{T})\,\operatorname{rank}(A_i^{*})
=
r_i^{2}.
\end{equation}
Hence
\begin{equation}\label{rank inequlity}
\operatorname{rank}\left(\sum_{i=1}^{r} A_i^{T}\otimes A_i^{*}\right)
\leq
\sum_{i=1}^{r} r_i^{2}.
\end{equation}
Therefore, if $\sum_{i=1}^{r} r_i^{2}\leq N(d)$, then
\begin{equation}
\operatorname{rank}\left(\sum_{i=1}^{r} A_i^{T}\otimes A_i^{*}\right)\leq N(d),
\end{equation}
which implies $\Phi$ is not phase retrievable.
\end{proof}

\begin{remark}\label{rank equility}
From Theorem~\ref{general rank}, it is also natural to ask when the estimate in Equation \ref{rank inequlity} is sharp.\\
\\
Indeed, suppose that $A_i=\sum_{k=1}^{r_i} |u_{i,k}\rangle\langle v_{i,k}|$ is a rank decomposition of each $A_i$, then
\begin{equation}
A_i^{T}\otimes A_i^{*}
=
\sum_{k,\ell=1}^{r_i}
\bigl(|\overline{v_{i,k}}\rangle\otimes |v_{i,\ell}\rangle\bigr)
\bigl(\langle \overline{u_{i,k}}|\otimes \langle u_{i,\ell}|\bigr).
\end{equation}
Hence
\begin{equation}
\sum_{i=1}^{r} A_i^{T}\otimes A_i^{*}
=
\sum_{i=1}^{r}\sum_{k,\ell=1}^{r_i}
\bigl(|\overline{v_{i,k}}\rangle\otimes |v_{i,\ell}\rangle\bigr)
\bigl(\langle \overline{u_{i,k}}|\otimes \langle u_{i,\ell}|\bigr).
\end{equation}
Therefore equality
\begin{equation}\label{rank equation}
\operatorname{rank}\left(\sum_{i=1}^{r} A_i^{T}\otimes A_i^{*}\right)
=
\sum_{i=1}^{r} r_i^{2}
\end{equation}
holds whenever the families $\{\,|v_{i,k}\rangle\langle v_{i,\ell}| : 1\leq i\leq r,\ 1\leq k,\ell\leq r_i\,\}$ and $\{\,|u_{i,k}\rangle\langle u_{i,\ell}| : 1\leq i\leq r,\ 1\leq k,\ell\leq r_i\,\}$ are linearly independent.
\end{remark}

\subsection{Important channel classes}

\noindent We next turn to entanglement-breaking channels, whose rank-one Kraus structure makes them a natural class in which phase retrievability tends to fail. 

\begin{definition}
A quantum channel $\Phi:B(\mathcal H)\to B(\mathcal K)$ is called \emph{entanglement breaking} if it admits a Kraus representation $\Phi_{\mathcal A}(\rho)=\sum_{i=1}^{r} A_i \rho A_i^{*}$ such that $\operatorname{rank}(A_i)=1,\; 1\leq i\leq r.$ The smallest such integer $r$ is called the \emph{entanglement-breaking rank} of $\Phi$, and is denoted by $\operatorname{ebr}(\Phi)$.
\end{definition}

\noindent For these channels, the previous Theorem \ref{general rank} becomes particularly simple.

\begin{corollary}\label{EB not PR}

Let $\Phi:B(\mathcal H)\to B(\mathcal K)$ be an entanglement-breaking channel, and let $\operatorname{ebr}(\Phi)$ denote the entanglement-breaking rank of $\Phi$. If $\operatorname{ebr}(\Phi)\leq N(d)$, then $\Phi$ is not phase retrievable. Moreover, if $\Phi_{\mathcal A}(\rho)=\sum_{i=1}^{\operatorname{ebr}(\Phi)} A_i \rho A_i^{*}$ is the rank-one Kraus representation of $\Phi$, where $A_i=|u_i\rangle\langle v_i|$ for $1\leq i\leq \operatorname{ebr}(\Phi)$, then the equality in Equation \ref{rank equation} holds whenever the families $\{\,|u_i\rangle\langle u_i| : 1\leq i\leq \operatorname{ebr}(\Phi)\,\}$ and $\{\,|v_i\rangle\langle v_i| : 1\leq i\leq \operatorname{ebr}(\Phi)\,\}$ are both linearly independent.
\end{corollary}

\noindent The proof of Corollary \ref{EB not PR} follows directly from Theorem~\ref{general rank} and the discussion in Remark \ref{rank equility}.\\
\\
\noindent We next consider twirling channels, an important class of quantum channels arising from unitary representations of compact groups. Their symmetry allows several associated quantities to be computed explicitly in terms of the irreducible multiplicities and dimensions, making them a natural setting for applying the general phase retrievability criterion above.\\
\\
In our previous work on twirling channels \cite{HanLiu2024}, phase retrievability was studied from the subspace point of view. More precisely, the phase-retrievability index $\operatorname{pr}(\Phi)$ is the largest dimension of a subspace on which $\Phi$ is phase retrievable. Here we instead consider phase retrievability on the whole input space.\\
\\
Let $G$ be a compact group, $\mu$ be the normalized Haar measure on $G$, and $\pi:G\to U(\mathcal H)$ be a finite-dimensional unitary representation. The twirling channel associated with $\pi$ is defined by
\begin{equation}
\Phi_{\pi}(\rho)=\int_G \pi(g)\rho \pi(g^{-1})\,d\mu(g).
\end{equation}
If $G$ is finite, then
\begin{equation}
\Phi_{\pi}(\rho)=\frac{1}{|G|}\sum_{g\in G}\pi(g)\rho \pi(g^{-1}).
\end{equation}

\noindent Suppose that $\pi$ can be decomposed as
\begin{equation}
U\pi(g)U^*=\bigoplus_{\alpha=1}^{s}\bigl(I_{m_\alpha}\otimes \pi_\alpha(g)\bigr),
\qquad g\in G,
\end{equation}
where $\pi_1,\dots,\pi_s$ are inequivalent irreducible representations, $m_\alpha$ is the multiplicity of $\pi_\alpha$, and $n_\alpha=\dim \mathcal H_\alpha$ for each $\alpha$. Then $d=\dim \mathcal H=\sum_{\alpha=1}^{s} m_\alpha n_\alpha.$\\
\\
We shall still use $\beta(\Phi_\pi)$ to denote the largest dimension of a quantum code of $\Phi_\pi$. It was also shown in \cite{HanLiu2024} that
\begin{equation}
\operatorname{Range}(\Phi_\pi)=A_\pi',
\end{equation}
where $A_\pi$ is the $C^*$-algebra generated by $\pi(G)$ and $A_\pi'$ is its commutant. By Schur's lemma and the above decomposition of $\pi$, one has
\begin{equation}
A_\pi' \cong \bigoplus_{\alpha=1}^{s}\bigl(M_{m_\alpha}(\mathbb C)\otimes I_{\mathcal H_\alpha}\bigr),
\end{equation}
and therefore
\begin{equation}
\dim A_\pi'=\sum_{\alpha=1}^{s} m_\alpha^2.
\end{equation}

\noindent The next theorem identifies the dimension of the operator system generated by the complementary channel of $\Phi_\pi$.

\begin{theorem}\label{twirling operator system}
Let $\Phi_\pi$ be the twirling channel associated with $\pi$. Then
\begin{equation}
\dim S_{\Phi_\pi^{c}}=\sum_{\alpha=1}^{s} m_\alpha^{2}.
\end{equation}
\end{theorem}

\begin{proof}
By Lemma~\ref{dim operator system}, we first have that
\begin{equation}
\dim S_{\Phi_\pi^{c}}=\dim \operatorname{Range}(\Phi_\pi^{*}).
\end{equation}
We further claim that $\Phi_\pi$ is self-adjoint with respect to the Hilbert--Schmidt inner product. Indeed, for any $X,Y\in B(\mathcal H)$,
\begin{equation}
\langle X,\Phi_\pi(Y)\rangle
=
\int_G \operatorname{Tr}\bigl(X^*\pi(g)Y\pi(g^{-1})\bigr)\,d\mu(g)
=
\int_G \operatorname{Tr}\bigl((\pi(g^{-1})X\pi(g))^*Y\bigr)\,d\mu(g)
=
\langle \Phi_\pi(X),Y\rangle,
\end{equation}
which shows that $\operatorname{Range}(\Phi_\pi^{*})=\operatorname{Range}(\Phi_\pi)$.\\
\\
Using the fact that $\operatorname{Range}(\Phi_\pi)=A_\pi'$, we obtain
\begin{equation}
\dim S_{\Phi_\pi^{c}}
=
\dim \operatorname{Range}(\Phi_\pi^{*})
=
\dim \operatorname{Range}(\Phi_\pi)
=
\dim A_\pi'
=
\sum_{\alpha=1}^{s} m_\alpha^{2}.
\end{equation}
This proves the result.
\end{proof}

\noindent Combining Theorem~\ref{twirling operator system} with Theorem~\ref{PR iff PSIC} and Theorem~\ref{lower bound HMW}, we have the following necessary condition immediately.

\begin{corollary}\label{twirling not PR}
Let $\Phi_\pi$ be the twirling channel associated with group representation $\pi$. If $\sum_{\alpha=1}^{s} m_\alpha^{2}\leq N(d)$, then $\Phi_\pi$ is not phase retrievable.
\end{corollary}

\noindent The present result and the twirling channel results in \cite{HanLiu2024} address the phase-retrievability problem from two complementary sides. In the previous paper, the main question was how large a subspace can remain phase retrievable for a given twirling channel $\Phi_\pi$. In the present paper, by contrast, the focus is on when the symmetry of $\Phi_\pi$ permits phase retrievability on the whole space $\mathcal H$.\\
\\
Thus the former concerns the size of phase-retrievable subspaces, while the latter gives a necessary condition for full-space phase retrievability. In particular, it was shown in \cite{HanLiu2024} that $\beta(\Phi_\pi)=\max_{\alpha} m_\alpha$ and $\operatorname{pr}(\Phi_\pi)\ge \max\left\{\beta(\Phi_\pi),\left\lfloor \frac d4\right\rfloor+1\right\},$ whereas Corollary~\ref{twirling not PR} shows that if $\sum_{\alpha=1}^{s} m_\alpha^2\le N(d),$ then $\Phi_\pi$ cannot be phase retrievable on the whole space $\mathcal H$.

\section{Quantum Interferometer and Frame Coupling}

\subsection{Quantum interference and operator-valued frame coupling}

We now turn from the structural results of Section~3 to the interferometric mechanism that can generate new directions in the associated operator system. Since the basic physical setup has already been illustrated in Figure~\ref{fig:interferometer}, we only present the ingredients needed for this section; see \cite{Oi03, OiAberg2006, Aberg04,BranciardEtAl2021,VanrietveldeChiribella2021} for related discussions of interference and coherent control of quantum channels.\\
\\
Let $\Phi_A(\rho)=\sum_{i=1}^{m}A_i\rho A_i^*$, and $\Phi_B(\rho)=\sum_{i=1}^{m}B_i\rho B_i^*,
\; \rho\in B(\mathcal H),$ be two quantum channels from $B(\mathcal H)$ to $B(\mathcal K)$. If the two sets of Kraus operators have different lengths, we pad the shorter one with zero operators so that both are indexed by $\{1,\dots,m\}$.\\
\\
Choose coherent Stinespring isometries $V_A$ and $V_B$ that are defined by
\begin{equation}
V_Ax=\sum_{i=1}^{m}A_ix\otimes |i\rangle,
\qquad
V_Bx=\sum_{i=1}^{m}B_ix\otimes |i\rangle,
\qquad x\in\mathcal H,
\end{equation}
with respect to a common environment basis $\{|i\rangle\}_{i=1}^{m}$. Let the path degree of freedom be recorded by $\{|0\rangle,|1\rangle\}$, and suppose the first beam splitter sends the input into an equal superposition of the two paths.
\begin{equation}
x\longmapsto \frac{1}{\sqrt2}\bigl(|0\rangle+|1\rangle\bigr)\otimes x.
\end{equation}

\noindent After the lower arm applies $V_A$, the upper arm applies $V_B$, and a phase $e^{i\theta}$ is inserted
in the upper arm, the state before the second beam splitter is now given by

\begin{equation}
\frac{1}{\sqrt2}
\Bigl(
|0\rangle\otimes V_Ax
+
e^{i\theta}|1\rangle\otimes V_Bx
\Bigr).
\end{equation}
Applying the second beam splitter, one obtains
\begin{equation}
\frac{1}{2}
\Bigl[
|0\rangle\otimes (V_A+e^{i\theta}V_B)x
+
|1\rangle\otimes (V_A-e^{i\theta}V_B)x
\Bigr].
\end{equation}
Hence the amplitude at port $0$ is governed by
\begin{equation}
\frac{1}{2}\bigl(V_A+e^{i\theta}V_B\bigr)
=
\sum_{i=1}^{m}\frac{1}{2}\bigl(A_i+e^{i\theta}B_i\bigr)\otimes |i\rangle.
\end{equation}
Therefore the corresponding port $0$ completely positive map $\Psi_\theta$ is given by
\begin{equation}\label{port cp map}
\Psi_\theta(\rho)
=
\sum_{i=1}^{m}\frac{1}{4}M_i(\theta)\rho M_i(\theta)^*,
\qquad
M_i(\theta):=A_i+e^{i\theta}B_i.
\end{equation}
The physical Kraus operators of the port map are $\frac{1}{2}M_i(\theta)$. Since multiplication by the nonzero scalar $\frac{1}{2}$ does not affect pure-state injectivity or the associated operator-system rank conditions, we call $M_i(\theta)$ the unnormalized port operators. These operators suggest a coordinatewise way to combine two operator-valued frames. We now introduce the following definition.
\begin{definition}\label{def:ovf-coupling}
Let $\mathcal A=\{A_i\}_{i=1}^{m}$ and $\mathcal B=\{B_i\}_{i=1}^{m}$ be operator-valued frames in $B(\mathcal H,\mathcal K)$, and $c=(c_1,\dots,c_m)\in \mathbb C^m$. The \emph{frame coupling} between $\mathcal{A}$ and $\mathcal{B}$ with coefficient vector $c$ is the set $\mathcal A\oplus_c\mathcal B :=\{A_i+c_iB_i\}_{i=1}^{m}.$\\
\\
If $c_i=\lambda$ for all $i$, we call $\mathcal A\oplus_c\mathcal B$ a
\emph{uniform coupling}, and if moreover $|\lambda|=1$, we call it a
\emph{uniform coherent coupling}. In particular, when $\lambda=e^{i\theta}$, the family $\mathcal M_\theta(\mathcal A,\mathcal B) := \{A_i+e^{i\theta}B_i\}_{i=1}^{m}$ is called the \emph{interferometric coupling} of $\mathcal A$ and $\mathcal B$.
\end{definition}

\begin{remark}
Similar types of combinations also appear in frame theory, for example in weighted $g$-frames, and multiplier-type constructions \cite{Sun2006,BalazsAntoineGrybos2010,StoevaBalazs2020}. In the present setting, the scalar unimodular case is especially relevant, since it is the one naturally realized by a two-path interferometer.
\end{remark}

\begin{remark}\label{remark:implementation}
For fixed channels $\Phi_A$ and $\Phi_B$, the family $\{M_i(\theta)\}_{i=1}^{m}$ depends on the chosen Kraus realizations of the two arms. Hence the interferometric coupling is sensitive to the physical realization of the channels, not merely to their input-output action \cite{Aberg04,BranciardEtAl2021,VanrietveldeChiribella2021}.
\end{remark}

\subsection{Necessary Conditions for Phase Retrievability of the Port Map}

Let $\{A_i\}_{i=1}^{m}$ and $\{B_i\}_{i=1}^{m}$ be the sets of Kraus operators for the quantum channels in the lower and upper arms of a quantum interferometer, respectively. Both sets define Parseval operator-valued frames, that is,
\begin{equation}
\sum_{i=1}^{m}A_i^*A_i=I_{\mathcal H},
\qquad
\sum_{i=1}^{m}B_i^*B_i=I_{\mathcal H}.
\end{equation}
Let $T:=\sum_{i=1}^{m}A_i^*B_i,$ we have the following criterion for the port map.

\begin{theorem}\label{thm:port-frame-criterion}
Let $\{M_i(\theta)\}_{i=1}^{m}$ denote the unnormalized port operator set that arising from the quantum interferometer, given by $M_i(\theta)=A_i+e^{i\theta}B_i, \; i=1,\dots,m.$ Further define $E_\theta:=\sum_{i=1}^{m}M_i(\theta)^*M_i(\theta).$ Then, we have
\begin{equation}\label{E theta}
E_\theta = 2I_{\mathcal H}+e^{i\theta}T+e^{-i\theta}T^*.
\end{equation}
Moreover, the following are equivalent:
\begin{enumerate}
    \item[(i)] $\{M_i(\theta)\}_{i=1}^{m}$ is an operator-valued frame;
    \item[(ii)] $E_\theta$ is invertible;
    \item[(iii)] $-e^{-i\theta}\notin \sigma(T)$.
\end{enumerate}
\end{theorem}
\begin{proof}
Since both $\{A_i\}_{i=1}^{m}$ and $\{B_i\}_{i=1}^{m}$ are Parseval frames,
\begin{align}
E_\theta
&=
\sum_{i=1}^{m}(A_i+e^{i\theta}B_i)^*(A_i+e^{i\theta}B_i) \\
&=
\sum_{i=1}^{m}A_i^*A_i
+
e^{i\theta}\sum_{i=1}^{m}A_i^*B_i
+
e^{-i\theta}\sum_{i=1}^{m}B_i^*A_i
+
\sum_{i=1}^{m}B_i^*B_i \\
&=
2I_{\mathcal H}+e^{i\theta}T+e^{-i\theta}T^*,
\end{align}
which proves \eqref{E theta}.\\
\\
Now for every $x\in \mathcal H$,
\begin{equation}
\sum_{i=1}^{m}\|M_i(\theta)x\|^2=\sum_{i=1}^{m}\langle M_i(\theta)x, M_i(\theta)x\rangle=\langle E_\theta x,x\rangle.
\end{equation}
Since $\mathcal H$ is finite-dimensional, the set $\{M_i(\theta)\}_{i=1}^{m}$ is an
operator-valued frame if and only if $E_\theta$ is positive definite, and then invertible. Thus (i) and (ii) are equivalent.\\
\\
Next we show the equivalence between (ii) and (iii). We first verify that the operator $T$ is a contraction. For $x,y\in\mathcal H$,
\begin{equation}
\langle Tx,y\rangle
=
\sum_{i=1}^{m}\langle A_i^*B_i x,y\rangle
=
\sum_{i=1}^{m}\langle B_i x,A_i y\rangle.
\end{equation}
By Cauchy--Schwarz and the Parseval identities, we have
\begin{equation}
|\langle Tx,y\rangle|
\le
\Bigl(\sum_{i=1}^{m}\|B_i x\|^2\Bigr)^{1/2}
\Bigl(\sum_{i=1}^{m}\|A_i y\|^2\Bigr)^{1/2}
=
\|x\|\,\|y\|.
\end{equation}
Therefore $\|T\|\le 1$.\\
\\
Now let $\|x\|=1$. Then
\begin{equation}
\langle E_\theta x,x\rangle
=
2+e^{i\theta}\langle Tx,x\rangle+e^{-i\theta}\overline{\langle Tx,x\rangle}
=
2+2\operatorname{Re}\bigl(e^{i\theta}\langle Tx,x\rangle\bigr).
\end{equation}
Since $\|T\|\le 1$, we have $|\langle Tx,x\rangle|\le 1$, and hence $\langle E_\theta x,x\rangle\ge 0.$ Then we have that $E_\theta$ is singular if and only if there exists a unit vector $x$ such that $\langle E_\theta x,x\rangle=0.$\\
\\
By the formula above, this is equivalent to $\operatorname{Re}\bigl(e^{i\theta}\langle Tx,x\rangle\bigr)=-1.$
Since $|e^{i\theta}\langle Tx,x\rangle|\le 1$, this happens if and only if $e^{i\theta}\langle Tx,x\rangle=-1,$
that is, $\langle Tx,x\rangle=-e^{-i\theta}.$
But then
\begin{equation}
1=|\langle Tx,x\rangle|\le \|Tx\|\le 1,
\end{equation}
which forces $Tx$ to be a scalar multiple of $x$.
Hence $Tx=-e^{-i\theta}x.$ Therefore $E_\theta$ is singular if and only if $-e^{-i\theta}\in \sigma(T)$. That completes the proof.
\end{proof}

\noindent Theorem~\ref{thm:port-frame-criterion} further implies that, for the port map to be phase retrievable, its Kraus set must at least form an operator-valued frame.

\begin{proposition}\label{singular, not PR}
For $\dim\mathcal{H}\geq 2$, let $\{M_i(\theta)\}_{i=1}^{m}$ be the unnormalized Kraus operators of the port map from a quantum interferometer, with $M_i(\theta)=A_i+e^{i\theta}B_i, \; i=1,\dots,m,$ and $E_\theta:=\sum_{i=1}^{m}M_i(\theta)^*M_i(\theta).$ If $E_\theta$ is singular, then $\Psi_\theta$ is not pure-state injective.
\end{proposition}

\begin{proof}
Since $E_\theta\ge 0$, there exists a unit vector $k\in\ker(E_\theta)$ such that
\begin{equation}
0=\langle E_\theta k,k\rangle=\sum_{i=1}^{m}\langle M_i(\theta)k, M_i(\theta)k\rangle=\sum_{i=1}^{m}\|M_i(\theta)k\|^2,
\end{equation}
and hence $M_i(\theta)k=0, \; i=1,\dots,m.$\\
\\
Now select a unit vector $u\in\mathcal H$ such that $u\perp k$, and choose $x=\frac{u+k}{\sqrt2}$, and $y=\frac{u-k}{\sqrt2}.$\\
\\
Then $x$ and $y$ are unit vectors with $\rho_x\neq \rho_y$.  $M_i(\theta)k=0$ implies
\begin{equation}
M_i(\theta)x =\frac{1}{\sqrt2}M_i(\theta)u=M_i(\theta)y,
\end{equation}
for every $i=1,\dots,m$.\\
\\
Therefore $\Psi_\theta(\rho_x)=\Psi_\theta(\rho_y)$, and $\Psi_\theta$ is not pure-state injective.
\end{proof}

\noindent The following qubit example shows explicitly how singularity of $E_\theta$ causes the port completely positive map to fail to be a operator-valued frame and pure-state injective.

\begin{example}\label{singular port map}
Let $\mathcal H=\mathcal K=\mathbb C^2$, and consider the unitary channels $\Phi_A(\rho)=\rho, \;\Phi_B(\rho)=Z\rho Z$, where $Z=|0\rangle\langle 0|-|1\rangle\langle 1|$.\\
\\
Thus we may take $A_1=I, \; B_1=Z.$ For $\theta=0$, the port Kraus family is
\begin{equation}
M_1(0)=A_1+B_1=I+Z=2|0\rangle\langle 0|,
\end{equation}
so $E_0=M_1(0)^*M_1(0)=4|0\rangle\langle 0|$ is singular. Hence $\{M_1(0)\}$ is not an operator-valued frame.\\
\\
Moreover, we have 
\begin{equation}
\Psi_0(\rho)=\frac{1}{4}M_1(0)\rho M_1(0)^*=|0\rangle\langle 0|\rho |0\rangle\langle 0|.
\end{equation}
Choose $|+\rangle=\frac{|0\rangle+|1\rangle}{\sqrt2}, \; |-\rangle=\frac{|0\rangle-|1\rangle}{\sqrt2},$ then $\rho_+\neq \rho_-$, but $\Psi_0(\rho_+)=\Psi_0(\rho_-)=\frac{1}{2}|0\rangle\langle 0|$. Thus $\Psi_0$ is not pure-state injective.
\end{example}

\noindent We now extend the discussion in Section~3 to general completely positive maps. The key point is that the proofs of Theorem~\ref{PR iff PSIC}, Lemma~\ref{dim operator system}, and the corresponding rank formula rely only on Kraus computations together with the vectorization identity, and do not use trace preservation. For CP maps, define the complementary CP map by the same canonical Kraus formula. Therefore these arguments apply directly to the port map $\Psi_\theta$. In particular, we can verify that
\begin{equation}\label{port psi psic}
\Psi_\theta \text{ is pure-state injective }
\iff
\Psi_\theta^c \text{ is pure-state informationally complete},
\end{equation}
and
\begin{equation}\label{port rank formula}
\dim S_{\Psi_\theta^c}
=
\dim \operatorname{Range}(\Psi_\theta^*)
=
\operatorname{rank}\!\left(\sum_{i=1}^{m}M_i(\theta)^T\otimes M_i(\theta)^*\right).
\end{equation}
\\
Now let $X\in B(\mathcal K)$. A direct computation gives
\begin{align}\label{port adjoint}
\Psi_\theta^*(X)
&=
\sum_{i=1}^{m}\frac{1}{4}M_i(\theta)^*XM_i(\theta) \\
&=\frac{1}{4}(
\Phi_A^*(X)+\Phi_B^*(X)
+e^{i\theta}\sum_{i=1}^{m}A_i^*XB_i
+e^{-i\theta}\sum_{i=1}^{m}B_i^*XA_i \notag ).
\end{align}
Thus the adjoint of the port map contains interferometric cross terms
\begin{equation}
\Gamma_{AB}(X):=\sum_{i=1}^{m}A_i^*XB_i,
\qquad
\Gamma_{BA}(X):=\sum_{i=1}^{m}B_i^*XA_i.
\end{equation}
Equivalently, let $\{f_\mu\}$ be an orthonormal basis of $\mathcal K$ and $E_{\nu\mu}=|f_\nu\rangle\langle f_\mu|$, then
\begin{equation}
S_{\Psi_\theta^c}
=
\operatorname{span}
\Bigl\{
\Phi_A^*(E_{\nu\mu})+\Phi_B^*(E_{\nu\mu})
+e^{i\theta}\Gamma_{AB}(E_{\nu\mu})
+e^{-i\theta}\Gamma_{BA}(E_{\nu\mu})
:
\mu,\nu
\Bigr\}.
\end{equation}
\\
The dimension of the complementary operator system of the port map is determined by the Kraus operators of the two arms through the expansion
\begin{align}\label{port rank formula expanded}
\sum_{i=1}^{m}M_i(\theta)^T\otimes M_i(\theta)^*
&=
\sum_{i=1}^{m}A_i^T\otimes A_i^*
+\sum_{i=1}^{m}B_i^T\otimes B_i^* \\
&\quad
+e^{-i\theta}\sum_{i=1}^{m}A_i^T\otimes B_i^*
+e^{i\theta}\sum_{i=1}^{m}B_i^T\otimes A_i^*. \notag
\end{align}
\\
This is the operator system form of the enhancement mechanism: the first two sums come from the
arm channels separately, while the last two are the new interferometric contributions.\\
\\
Consequently, the necessary condition from Section~3 extends directly to the port map: if
\begin{equation}
\operatorname{rank}\!\left(\sum_{i=1}^{m}M_i(\theta)^T\otimes M_i(\theta)^*\right)\leq N(d),
\qquad d=\dim\mathcal H,
\end{equation}
then $\Psi_\theta$ is not pure-state injective, and thus not phase retrievable. More generally, every discussion in Section~3 whose proof depends only on Kraus representations extends to the port map in the same way.

\subsection{Classical mixing versus interferometric coupling}

A quantum interferometer provides a physical realization of the operator-valued frame coupling introduced above, by recombining the associated Kraus operators of the two arm channels. There are, however, other natural ways to combine two frames. Let $\{A_i\}_{i=1}^{m}$ and $\{B_i\}_{i=1}^{m}$ be Kraus operators of two quantum channels $\Phi_A$ and $\Phi_B$, and hence Parseval operator-valued frames. A standard alternative way to combine them is by weighted union.
\begin{equation}
\mathcal A\oplus_p\mathcal B := \{\sqrt p\,A_i\}_{i=1}^{m}\cup \{\sqrt{1-p}\,B_i\}_{i=1}^{m},
\; 0\le p\le 1.
\end{equation}
The associated channel is the classical mixture: $\Gamma_p(\rho)=p\Phi_A(\rho)+(1-p)\Phi_B(\rho)$.\\
\\
 Physically, $\Gamma_p$ corresponds to selecting one arm by a classical random switch, or more generally by a classically controlled device, rather than coherently recombining the two paths \cite{Yuan2018,Wechs2021,Dong2019}.\\
 \\
The essential difference from interferometric coupling can already be seen from the adjoint maps. Indeed, for classical mixing, we have
\begin{equation}
\Gamma_p^*(X)=p\Phi_A^*(X)+(1-p)\Phi_B^*(X),
\end{equation}
and hence $S_{\Gamma_p^c}= \operatorname{Range}(\Gamma_p^*) \subseteq S_{\Phi_A^c}+S_{\Phi_B^c}.$\\
\\
Thus classical mixing cannot enlarge the complementary operator system outside those directions already present in the two arm channels.\\
\\
By contrast, the interferometric port map satisfies
\begin{equation}
\Psi_\theta^*(X)
=\frac{1}{4}
(\Phi_A^*(X)+\Phi_B^*(X)
+e^{i\theta}\Gamma_{AB}(X)
+e^{-i\theta}\Gamma_{BA}(X)),
\end{equation}
so in addition to the two arm contributions it contains the cross terms
\begin{equation}
\Gamma_{AB}(X):=\sum_{i=1}^{m}A_i^*XB_i,
\qquad
\Gamma_{BA}(X):=\sum_{i=1}^{m}B_i^*XA_i.
\end{equation}
Therefore, if either cross term space escapes $S_{\Phi_A^c}+S_{\Phi_B^c}$, coherent recombination
may enlarge the complementary operator system, something classical mixing cannot achieve.\\
\\
Thus, the discussion of this subsection is that classical mixing and interferometric coupling combine the two arm channels in fundamentally different ways. Classical mixing remains within the operator system directions already contributed by the two arms, whereas interferometric coupling introduces additional cross terms through coherent recombination. It is precisely these new terms that make interferometric coupling the more effective mechanism for enhancing phase retrieval behavior. Section~5 will make this mechanism quantitative through the injectivity index and some explicit examples.

\section{Quantitative Index and Illustrative Examples}

\subsection{Injectivity indices}

The results of Sections~3 and~4 describe the injectivity properties of quantum channels and of the port maps arising from quantum interferometry. To study these objects quantitatively, it is useful to introduce indices that measure how strongly they separate quantum states. This viewpoint is also consistent with recent work on irreversibility, information loss, and incompatibility in quantum information theory \cite{LuoSun2024,LiLuoSunWang2025,GuoLuoZhang2025,SunLi2024}.\\
\\
Let $D(\mathcal H):=\{\rho\in B(\mathcal H): \rho\ge 0,\ \operatorname{Tr}(\rho)=1\}$ denote the set of quantum states on a finite-dimensional Hilbert space $\mathcal H$, and let $H_0(\mathcal H):=\{X\in B(\mathcal H): X=X^*,\ \operatorname{Tr}(X)=0\}$
be the traceless Hermitian subspace. 
The space $H_0(\mathcal H)$ is the natural space for studying injectivity on states, because differences of states belong to $H_0(\mathcal H)$, and conversely every nonzero element of $H_0(\mathcal H)$ is, up to a positive scalar, the difference of two distinct states.\\
\\
This motivates the following definition.

\begin{definition}\label{def:CP_injectivity_index}
Let $\Phi:B(\mathcal H)\to B(\mathcal K)$ be a completely positive map. 
The \emph{CP injectivity} of $\Phi$ is defined by
\begin{equation}
I(\Phi):=
\inf_{X\in H_0(\mathcal H)\setminus\{0\}}
\frac{\|\Phi(X)\|_2}{\|X\|_2}.
\end{equation}
\end{definition}

\noindent The next proposition gives the basic interpretation of this index.

\begin{theorem}\label{prop:I_state_injective}
Let $\Phi:B(\mathcal H)\to B(\mathcal K)$ be a completely positive map. Then $I(\Phi)\ge 0,$ and the following are equivalent:
\begin{enumerate}
\item[(i)] $I(\Phi)>0$;
\item[(ii)] $\ker(\Phi)\cap H_0(\mathcal H)=\{0\}$;
\item[(iii)] $\Phi$ is injective on $D(\mathcal H)$.
\end{enumerate}
\end{theorem}

\begin{proof}
The inequality $I(\Phi)\ge 0$ is immediate from the definition.\\
\\
We first prove the equivalence of (i) and (ii). 
If there exists a nonzero $X\in H_0(\mathcal H)$ such that $\Phi(X)=0$, then $\frac{\|\Phi(X)\|_2}{\|X\|_2}=0,$ and hence $I(\Phi)=0$. \\
\\
Conversely, assume that $\ker(\Phi)\cap H_0(\mathcal H)=\{0\},$ and consider the unit sphere $S:=\{X\in H_0(\mathcal H): \|X\|_2=1\}.$ Since $H_0(\mathcal H)$ is finite dimensional, and the map $X\mapsto \|\Phi(X)\|_2$ is continuous on $S$, it attains its minimum at some $X_0\in S$. By the assumption, this minimum cannot be zero. Therefore
\begin{equation}
I(\Phi)=\min_{X\in S}\|\Phi(X)\|_2>0.
\end{equation}

\noindent Next we prove that (ii) implies (iii). 
Suppose that there exist two states $\rho_1,\rho_2\in D(\mathcal H)$ such that $\Phi(\rho_1)=\Phi(\rho_2)$. Then we have $\Phi(\rho_1-\rho_2)=0.$ Since $\rho_1-\rho_2\in H_0(\mathcal H)$, condition (ii) implies $\rho_1=\rho_2,$ hence $\Phi$ is injective on $D(\mathcal H)$.\\
\\
Finally, we prove that (iii) implies (ii). 
Suppose that there exists a nonzero $X\in H_0(\mathcal H)$ such that $\Phi(X)=0.$ Then there exist distinct states $\rho_1,\rho_2\in D(\mathcal H)$ and $t>0$ such that $X=t(\rho_1-\rho_2).$ Applying $\Phi$ gives
\begin{equation}
0=\Phi(X)=t\bigl(\Phi(\rho_1)-\Phi(\rho_2)\bigr),
\end{equation}
and then $\Phi(\rho_1)=\Phi(\rho_2)$, which contradicts the injectivity of $\Phi$.
\end{proof}

\noindent For qubit systems, this condition also agrees with the pure-state injectivity discussed earlier.

\begin{corollary}\label{cor:qubit_pure_state_equiv}
Let $\Phi:B(\mathbb C^2)\to B(\mathcal K)$ be a completely positive map. Then
\begin{equation}
I(\Phi)>0
\quad\Longleftrightarrow\quad
\Phi \text{ is pure-state injective.}
\end{equation}
\end{corollary}

\begin{proof}
If $I(\Phi)>0$, then Proposition~\ref{prop:I_state_injective} shows that $\Phi$ is injective on all states, hence in particular on pure states.\\
\\
Conversely, assume that $\Phi$ separates pure states and suppose that there exists $0\neq X\in H_0(\mathbb C^2)$ satisfying $\Phi(X)=0.$
Since $X$ is a nonzero traceless Hermitian $2\times 2$ matrix, its eigenvalues are $\lambda$ and $-\lambda$ for some $\lambda>0$. 
Let $u,v$ be corresponding orthonormal eigenvectors. Then $X=\lambda(\rho_u-\rho_v),$ where $\rho_u,\rho_v$ are pure states. Applying $\Phi$ gives
\begin{equation}
0=\Phi(X)=\lambda\bigl(\Phi(\rho_u)-\Phi(\rho_v)\bigr),
\end{equation}
hence $\Phi(\rho_u)=\Phi(\rho_v),$ which implies $\rho_u=\rho_v$ by assumption and that is impossible. Therefore
\begin{equation}
\ker(\Phi)\cap H_0(\mathbb C^2)=\{0\},
\end{equation}
and Proposition~\ref{prop:I_state_injective} yields $I(\Phi)>0$.
\end{proof}

\noindent We now derive a computable formula for the CP injectivity index $I(\Phi)$. Since completely positive maps preserve Hermitian operators, we may choose an orthonormal basis $\{F_i\}_{i=1}^{d^2-1}$ of $H_0(\mathcal H)$ and an orthonormal basis $\{G_\mu\}_{\mu=0}^{n^2-1}$ of the real Hilbert space of Hermitian operators on $\mathcal K$, where $n=\dim\mathcal K$. Define
\begin{equation}
\mathcal T_\Phi=\bigl(t_{\mu i}\bigr)_{0\le \mu\le n^2-1,\ 1\le i\le d^2-1},
\qquad
t_{\mu i}:=\langle G_\mu,\Phi(F_i)\rangle,
\end{equation}
and then we have the following.

\begin{proposition}\label{prop:I_matrix_formula}
With the notation above, for a completely positive map $\Phi$, we have
\begin{equation}
I(\Phi)=\sqrt{\lambda_{\min}\bigl(\mathcal T_\Phi^*\mathcal T_\Phi\bigr)}.
\end{equation}
\end{proposition}

\begin{proof}
For any $X\in H_0(\mathcal H)$, we have $X=\sum_{i=1}^{d^2-1}x_iF_i.$ Since $\{F_i\}$ is orthonormal,
\begin{equation}
\|X\|_2^2=\sum_{i=1}^{d^2-1}x_i^2.
\end{equation}
Also,
\begin{equation}
	\Phi(X)=\sum_{i=1}^{d^2-1}x_i\Phi(F_i) =\sum_{\mu=0}^{n^2-1}\left(\sum_{i=1}^{d^2-1}t_{\mu i}x_i\right)G_\mu.
\end{equation}
Because $\{G_\mu\}$ is also orthonormal, it follows that $\|\Phi(X)\|_2^2=\|\mathcal T_\Phi x\|_2^2.$\\
\\
Therefore
\begin{equation}
I(\Phi)^2 =\inf_{x\neq 0}\frac{\|\mathcal T_\Phi x\|_2^2}{\|x\|_2^2} =
\inf_{x\neq 0}\frac{\langle x,\mathcal T_\Phi^*\mathcal T_\Phi x\rangle}{\langle x,x\rangle}.
\end{equation}
Since $\mathcal T_\Phi^*\mathcal T_\Phi$ is positive semidefinite, we have $I(\Phi)^2=\lambda_{\min}\bigl(\mathcal T_\Phi^*\mathcal T_\Phi\bigr),$ which proves the claim after taking square roots.
\end{proof}

\noindent We next summarize several basic properties of $I(\Phi)$.

\begin{proposition}\label{prop:I_properties}
Let $\Phi,\Psi:B(\mathcal H)\to B(\mathcal K)$ be completely positive maps. Then the following hold.

\begin{enumerate}
\item[(i)] \textbf{Positive homogeneity.} For every scalar $c\ge 0$, $I(c\Phi)=c\,I(\Phi).$

\item[(ii)] \textbf{Unitary invariance.} If $\mathcal U(X)=UXU^*$ and $\mathcal V(Y)=VYV^*$ are unitary conjugation maps, then $I(\mathcal V\circ \Phi\circ \mathcal U)=I(\Phi).$

\item[(iii)] \textbf{Continuity.} Let
\begin{equation}
\|\Xi\|_{0\to 2}:=
\sup_{X\in H_0(\mathcal H)\setminus\{0\}}
\frac{\|\Xi(X)\|_2}{\|X\|_2},
\end{equation}
then we have $|I(\Phi)-I(\Psi)|\le \|\Phi-\Psi\|_{0\to 2}.$
\end{enumerate}
\end{proposition}

\begin{proof}
Property (i) is immediate from the definition.\\
\\
For (ii), note that $\mathcal V$ preserves the Hilbert--Schmidt norm on the output space, then
\begin{equation}
\frac{\|(\mathcal V\circ\Phi\circ\mathcal U)(X)\|_2}{\|X\|_2}
=
\frac{\|\Phi(\mathcal U(X))\|_2}{\|\mathcal U(X)\|_2}.
\end{equation}
Taking the infimum over $X\in H_0(\mathcal H)\setminus\{0\}$ will prove the claim.\\
\\
For (iii), let $L=\|\Phi-\Psi\|_{0\to 2}.$
Then for every $X\in H_0(\mathcal H)$ with $\|X\|_2=1$,
\begin{equation}
\|\Phi(X)\|_2
\ge
\|\Psi(X)\|_2-\|(\Phi-\Psi)(X)\|_2
\ge
I(\Psi)-L.
\end{equation}
Taking the infimum over such $X$ gives $I(\Phi)\ge I(\Psi)-L.$ Interchanging $\Phi$ and $\Psi$ yields the reverse inequality and then completes the proof.
\end{proof}

\noindent To measure the overall performance of a completely positive map with respect to state injectivity, it is also natural to introduce an averaged quantity.

\begin{definition}\label{def:Iav}
Let $\Phi:B(\mathcal H)\to B(\mathcal K)$ be a completely positive map. 
We define the \emph{average CP injectivity} of $\Phi$ by
\begin{equation}
I_{\mathrm{av}}(\Phi):=
\frac{\|\mathcal T_\Phi\|_F}{\sqrt{d^2-1}},
\end{equation}
where $\|\cdot\|_F$ denotes the Frobenius norm. Equivalently,
\begin{equation}
I_{\mathrm{av}}(\Phi)
=
\left(
\frac{\operatorname{Tr}\bigl(\mathcal T_\Phi^*\mathcal T_\Phi\bigr)}{d^2-1}
\right)^{1/2}
=
\left(
\frac{1}{d^2-1}\sum_{i=1}^{d^2-1}\|\Phi(F_i)\|_2^2
\right)^{1/2}.
\end{equation}
\end{definition}
\noindent The following property is then immediate since $\|\Phi\|_{0\to 2}^2 = \lambda_{\max}\bigl(\mathcal T_\Phi^*\mathcal T_\Phi\bigr).$

\begin{proposition}\label{prop:I_Iav_relation}
The quantity $I_{\mathrm{av}}(\Phi)$ is independent of the choice of orthonormal basis of $H_0(\mathcal H)$ and satisfies
\begin{equation}
I(\Phi)\le I_{\mathrm{av}}(\Phi)\le \|\Phi\|_{0\to 2}.
\end{equation}
\end{proposition}

\noindent The indices $I(\Phi)$ and $I_{\mathrm{av}}(\Phi)$ are not introduced merely as abstract quantities. Like other quantitative indices for quantum processes, they provide a concrete way to detect degeneracy, compare parameter regimes, and here help to analyze the behavior of the completely positive maps arising from interferometric coupling \cite{LuoSun2024,LiLuoSunWang2025,GuoLuoZhang2025}. In the next subsection, we discuss several concrete examples.

\subsection{Illustrative examples}

We now discuss several interferometric examples. For each pair of arm channels, let $\Psi_\theta$ denote the corresponding port completely positive map generated by the quantum interferometer. The heat maps are shown as functions of the phase $\theta$ and the parameter of the chosen arm families. We focus on the indices $I(\Psi_\theta)$ and $I_{\mathrm{av}}(\Psi_\theta)$ due to the previous discussion. In all figures, darker regions indicate weaker state separation, with the darkest regions corresponding to degeneracy where the index vanishes, while brighter regions indicate stronger separation and therefore better phase retrieval behavior.\\
\\

\begin{example}
We first consider the interference between a qubit amplitude damping channel
and a complete \(Z\)-dephasing channel. The first arm is the amplitude damping
channel with Kraus operators
\begin{equation}
A_0=|0\rangle\langle 0|+\sqrt{1-\gamma}\,|1\rangle\langle 1|,
\qquad
A_1=\sqrt{\gamma}\,|0\rangle\langle 1|,
\end{equation}
where \(0\leq \gamma\leq 1\). For \(0\leq \gamma<1\), this channel is phase
retrievable. The second arm is the complete \(Z\)-dephasing channel with Kraus
operators
\begin{equation}
B_0=|0\rangle\langle 0|,
\qquad
B_1=|1\rangle\langle 1|.
\end{equation}
This channel is not phase retrievable.\\
\\
The corresponding heat maps are shown in Figure~\ref{fig:ex-ad-zdephase}.
They show that the indices depend strongly on the interferometric phase.
In particular, the port map loses injectivity near the destructive-interference
phase, while it is pure-state injective for most parameter values away from
that region.
\end{example}

\begin{figure}[H]
\centering
\includegraphics[width=0.495\textwidth]{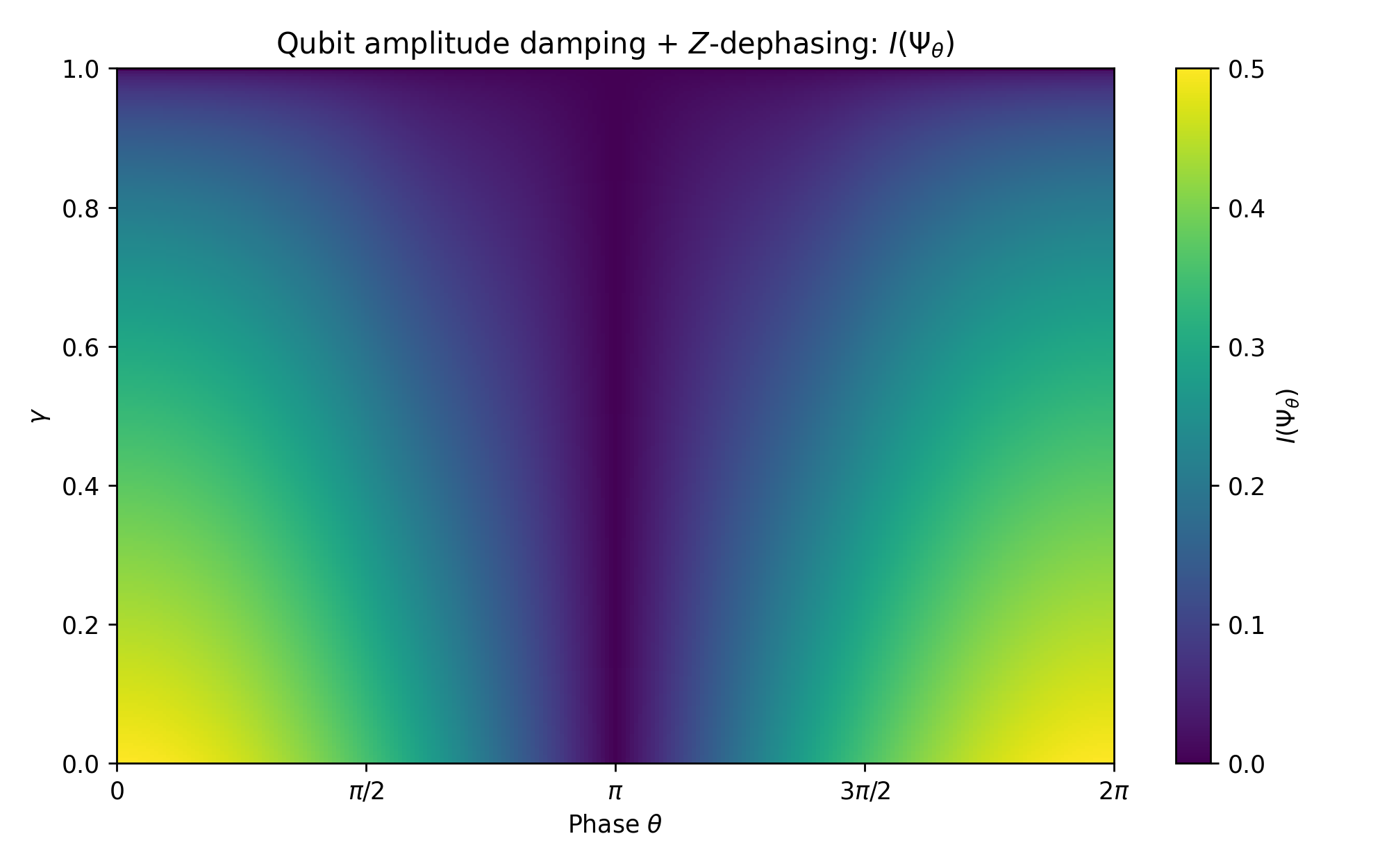}
\hfill
\includegraphics[width=0.495\textwidth]{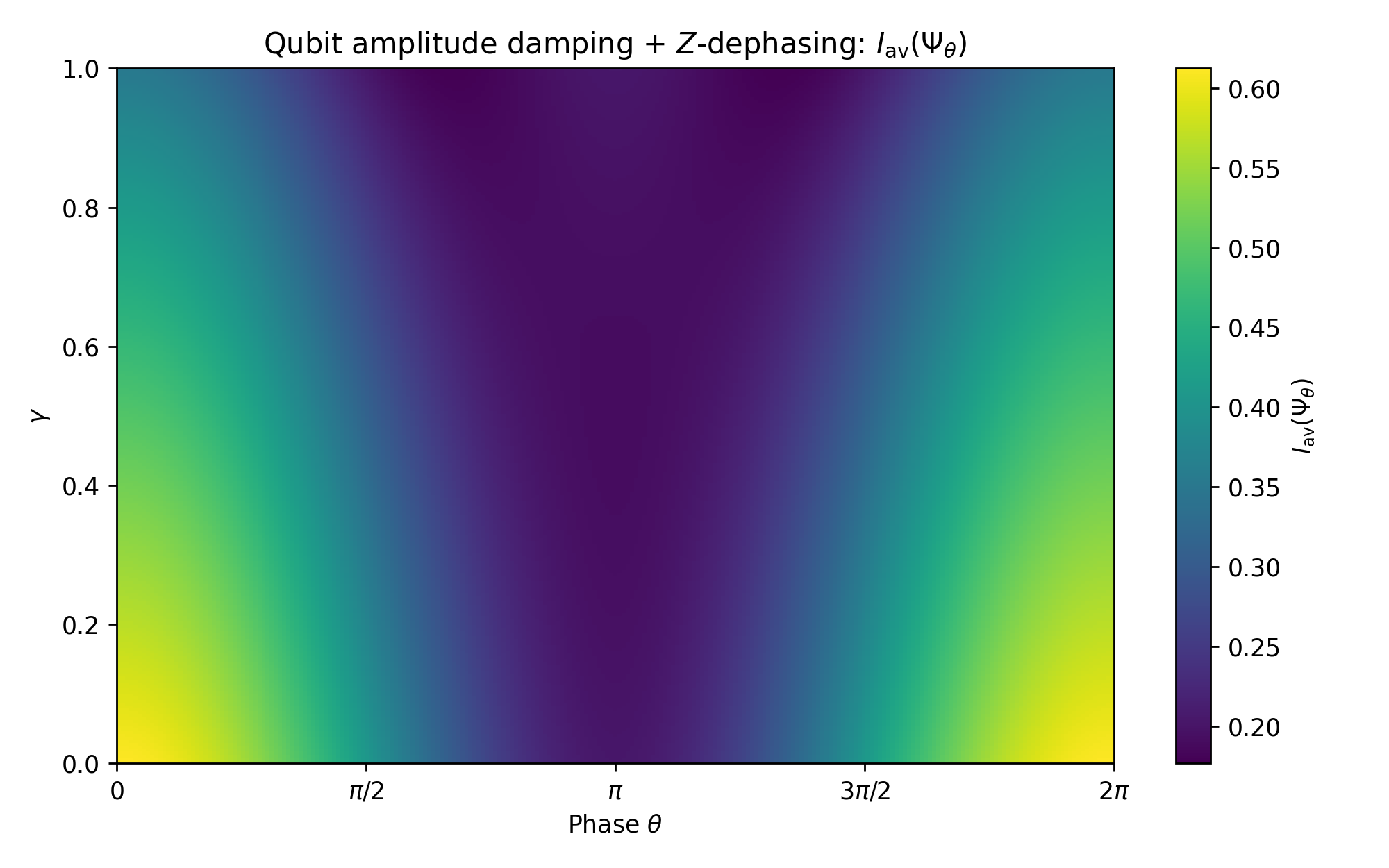}
\caption{Heat maps of $I(\Psi_\theta)$ (left) and $I_{\mathrm{av}}(\Psi_\theta)$ (right) for the interferometric coupling of qubit amplitude damping and $Z$-dephasing.}
\label{fig:ex-ad-zdephase}
\end{figure}
\FloatBarrier

\noindent For the channels considered in the next two examples, each arm is individually not phase retrievable. We show that, after interferometric coupling, the resulting port map becomes phase retrievable for a large range of parameter values, while also exhibiting relatively low compression.\\

\begin{example}[Reset channel and rotated complete dephasing]
We next consider the following two families of arm channels, given by their Kraus operators:
\begin{equation}
A_0=|0\rangle\langle 0|,
\qquad
A_1=|0\rangle\langle 1|,
\end{equation}
and
\begin{equation}
|b_0(\alpha)\rangle
=
\cos\frac{\alpha}{2}|0\rangle+\sin\frac{\alpha}{2}|1\rangle,
\qquad
|b_1(\alpha)\rangle
=
-\sin\frac{\alpha}{2}|0\rangle+\cos\frac{\alpha}{2}|1\rangle.
\end{equation}
The second arm is the complete dephasing channel in the rotated basis
\(\{|b_0(\alpha)\rangle,|b_1(\alpha)\rangle\}\), with Kraus operators
\begin{equation}
B_0(\alpha)=|b_0(\alpha)\rangle\langle b_0(\alpha)|,
\qquad
B_1(\alpha)=|b_1(\alpha)\rangle\langle b_1(\alpha)|.
\end{equation}
Both arms are individually not phase retrievable. Indeed, the first arm is a
reset channel sending every input state to \(|0\rangle\langle 0|\), while the
second arm is a complete dephasing channel and therefore loses the relative
phase in the rotated basis.\\
\\
Figure~\ref{fig:ex-zx-dephase} shows that, although both arms are individually
not phase retrievable, the interferometric port map has a large region where
the injectivity index is positive.
\end{example}

\begin{figure}[H]
\centering
\includegraphics[width=0.495\textwidth]{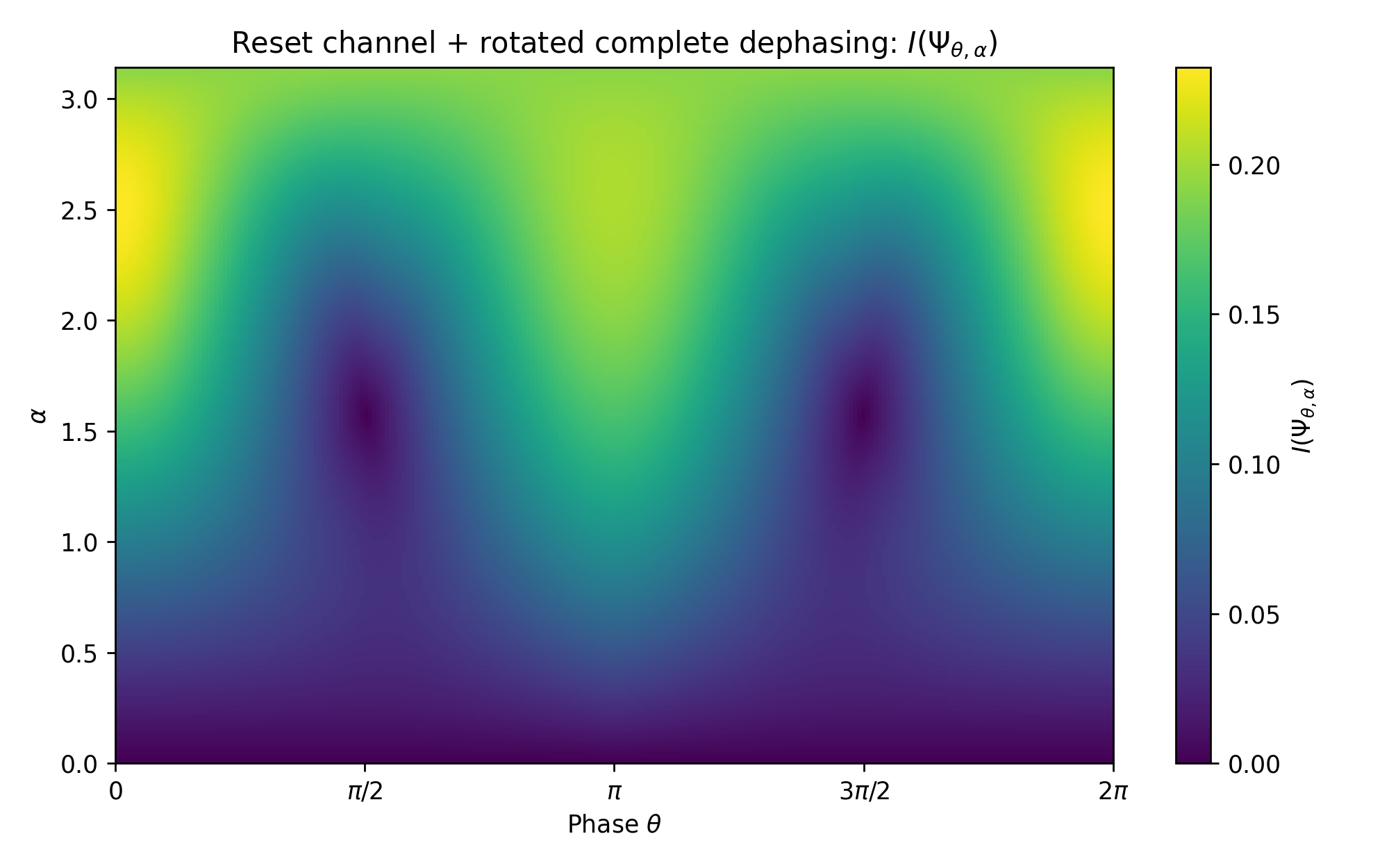}
\hfill
\includegraphics[width=0.495\textwidth]{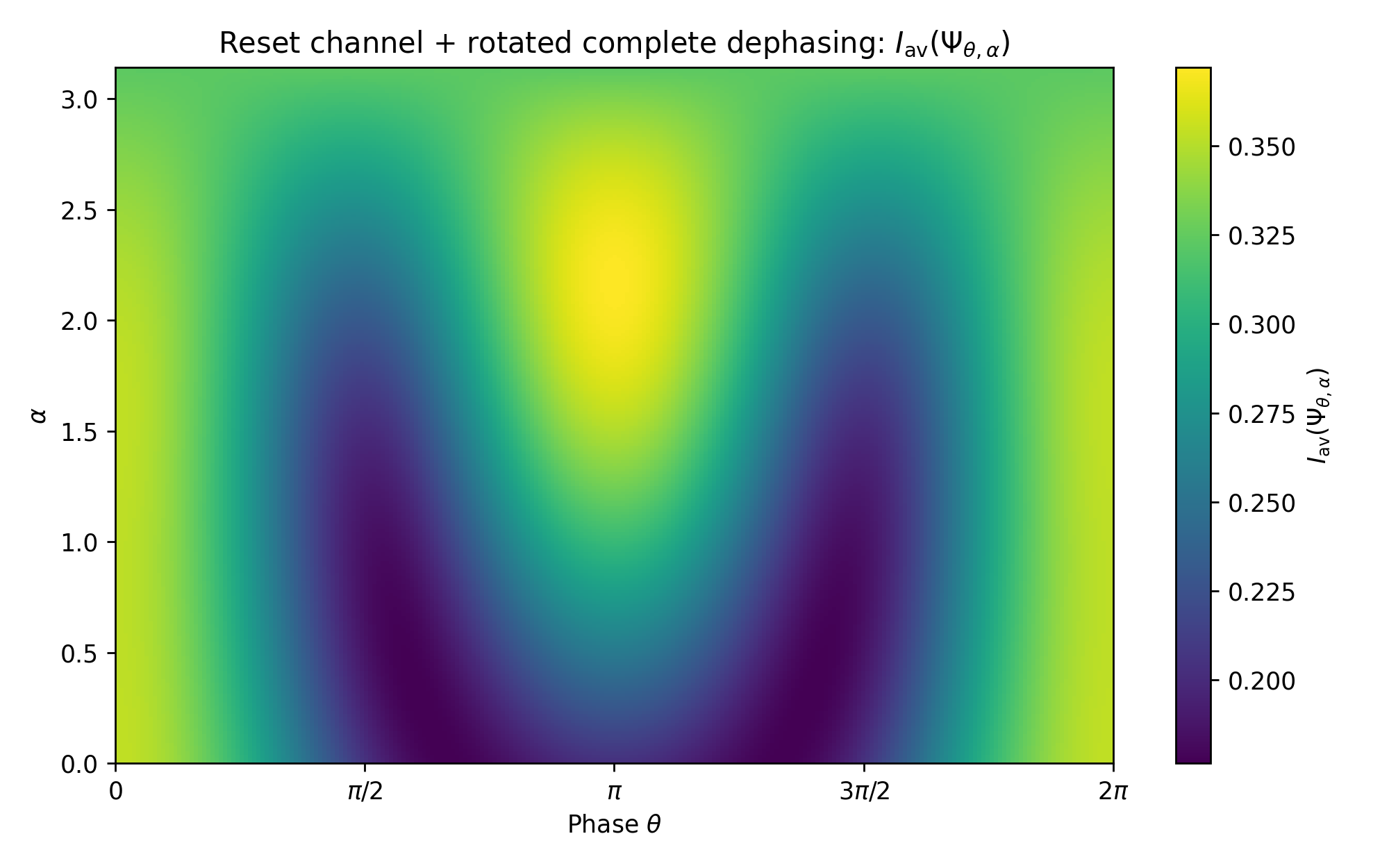}
\caption{Heat maps of $I(\Psi_{\theta,\alpha})$ (left) and $I_{\mathrm{av}}(\Psi_{\theta,\alpha})$ (right) for the interferometric coupling of the reset channel and the rotated complete dephasing channel.}
\label{fig:ex-zx-dephase}
\end{figure}
\FloatBarrier

\begin{example}[Rotated trine families]
We now consider a pair of qubit channels generated by trine families. Let
\begin{equation}
|v_k\rangle=\frac{1}{\sqrt{2}}\bigl(|0\rangle+\omega^k|1\rangle\bigr),
\qquad
\omega=e^{2\pi i/3},
\qquad
k=0,1,2,
\end{equation}
and define
\begin{equation}
R_y(\alpha)=e^{-i\alpha \sigma_y/2},
\qquad
|v_k(\alpha)\rangle=R_y(\alpha)|v_k\rangle.
\end{equation}
The two arms are given by the Kraus operators
\begin{equation}
A_k(\alpha)=\sqrt{\frac{2}{3}}\,|v_k(\alpha)\rangle\langle v_k(\alpha)|,
\qquad
B_k=\sqrt{\frac{2}{3}}\,|v_k\rangle\langle v_k|,
\qquad
k=0,1,2.
\end{equation}
Each arm is individually not phase retrievable since each trine family
spans only a two-dimensional affine slice of the qubit Bloch sphere and does
not distinguish all pure states.\\
\\
The heat maps are shown in Figure~\ref{fig:ex-trine-trine}. Compared with the
dephasing examples, the low-index regions for $I(\Psi_{\theta,\alpha})$ have a more
curved structure, while \(I_{\mathrm{av}}(\Psi_{\theta,\alpha})\) gives a smoother
average measure of the same parameter dependence. This example shows that the
effect of interferometric coupling is not restricted to dephasing channels.
\end{example}

\begin{figure}[H]
\centering
\includegraphics[width=0.495\textwidth]{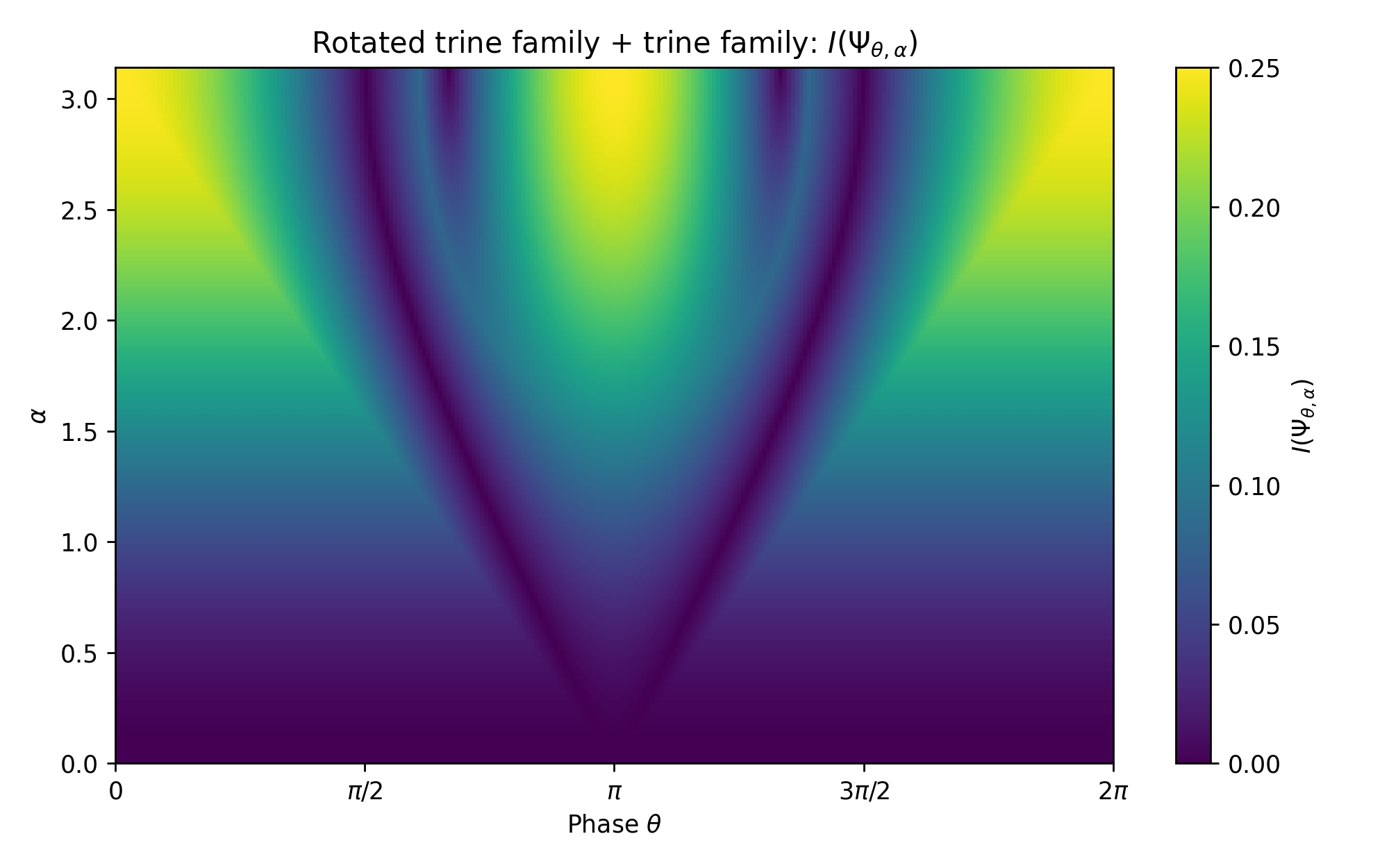}
\hfill
\includegraphics[width=0.495\textwidth]{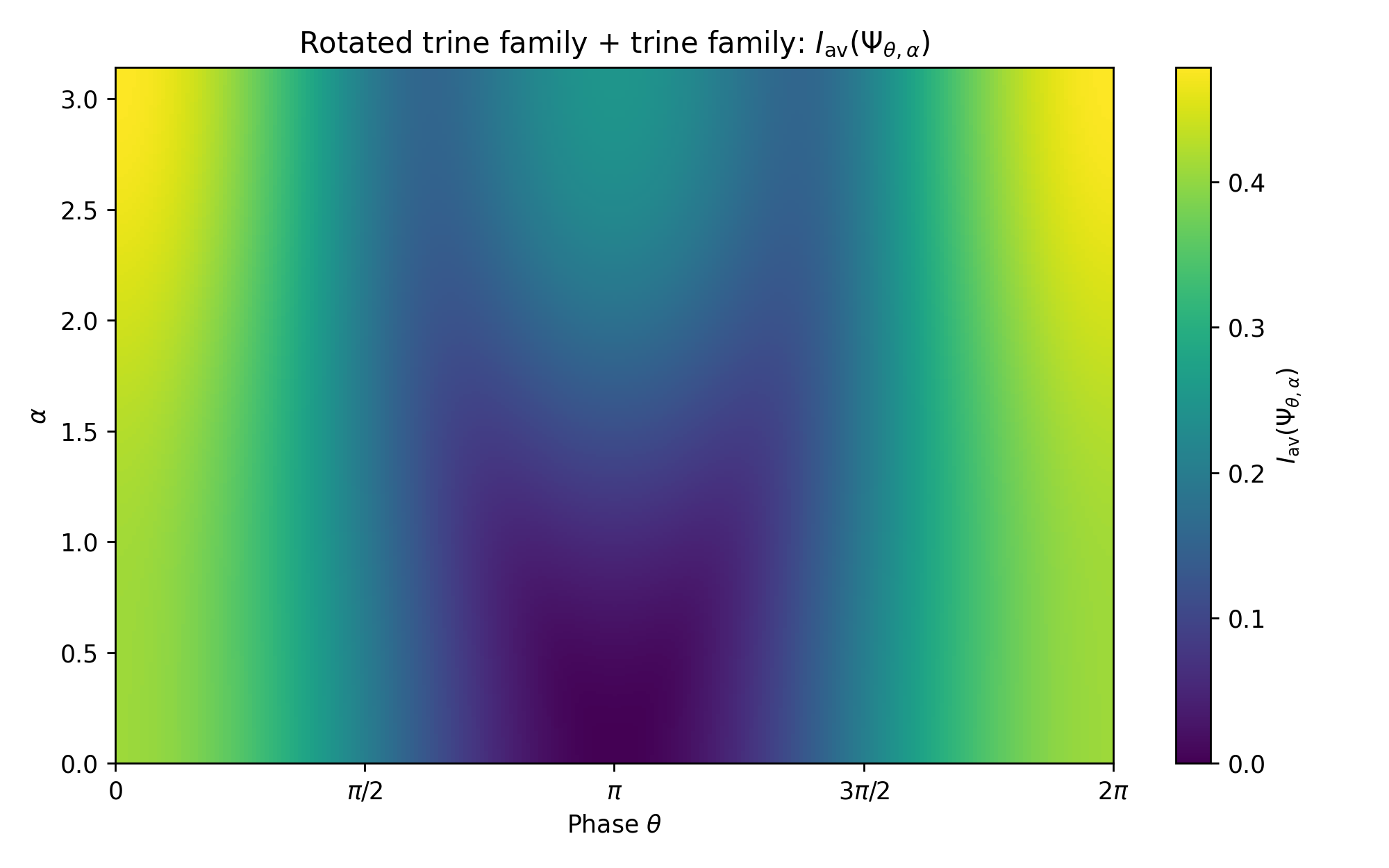}
\caption{Heat maps of $I(\Psi_{\theta,\alpha})$ (left) and $I_{\mathrm{av}}(\Psi_{\theta,\alpha})$ (right) for the interferometric coupling of the rotated trine family and the trine family.}
\label{fig:ex-trine-trine}
\end{figure}
\FloatBarrier

\noindent These examples support a common conclusion. The phase retrieval behavior of the port map depends strongly on the interferometric phase, and coherent coupling can substantially improve the behavior of the resulting completely positive map.

\section{Summary and outlook}

\noindent In this paper we studied phase retrievability of quantum channels through complementary channels, 
operator systems, and interferometric realization. We showed that phase retrievability is equivalent 
to pure-state informational completeness of the complementary channel, which turns the retrieval 
problem into a structural question about a complementary operator system. This viewpoint yields 
obstructions to phase retrievability, including criteria based on the dimension of operator systems and consequences for 
entanglement-breaking and twirling channels.\\
\\
The main novelty of the paper is the use of quantum interferometry as an enhancement mechanism, supported by the mathematical framework of frame theory. By coherently coupling two arm channels at the Kraus level, the interferometer generates cross terms that can enlarge the complementary operator system. In this way, the paper connects complementary channel methods from quantum information, operator-valued phase retrieval from frame theory, and coherent physical realization through quantum interferometry. The injectivity index and the examples provide concrete evidence for this enhancement principle.\\
\\
Several questions remain open. It would be interesting to identify general sufficient conditions under which interferometric coupling guarantees phase retrievability, or more generally enlarges the complementary operator system in a controlled way. It is also natural to study the role of the phase parameter $\theta$ more systematically, especially in higher dimensions, and to revisit special channel classes such as Schur, mixed unitary, and covariant channels from this interferometric perspective.


\section*{Statements and Declarations}

\noindent \textbf{Competing Interests.}
The authors declare no competing interests.\\
\\
\noindent \textbf{Data availability statement:}
No new data were created or analyzed in this study.\\
\\
\noindent \textbf{Funding statement:}
This research received no external funding.\\
\\
\noindent \textbf{Author Contributions.}
All authors contributed equally to the conception of the study, the analysis, and the writing of the manuscript. All authors reviewed and approved the final version of the manuscript.\\




\end{document}